\documentclass[11pt]{amsart}
\usepackage{geometry}
\usepackage[colorlinks,pagebackref,hypertexnames=false]{hyperref}
\usepackage{microtype}
\usepackage{xcolor}
\usepackage{amsmath}
\usepackage{amstext}
\usepackage{amsfonts}
\usepackage{amssymb}
\usepackage{amsbsy}
\usepackage{latexsym}
\usepackage[T1]{fontenc}
\usepackage[all]{xy}
\usepackage{tikz-cd}
\usepackage{hhline}
\usepackage{enumerate}
\usepackage{enumitem}
\usepackage[nameinlink,noabbrev]{cleveref}
\usepackage{autonum}
\usepackage{dsfont}
\usepackage{booktabs} 
\usepackage{caption} 
\usepackage{thmtools}
\allowdisplaybreaks

\hypersetup{
	colorlinks=true,
	linkcolor=red,
	citecolor=blue,
	filecolor=blue,
	urlcolor=blue
}

{}              

\newcounter{num}[section] %

\newtheorem{Mtheorem}{Theorem}

\newenvironment{theo}
{\refstepcounter{num}%
	\bigskip\noindent{\bf Theorem~\arabic{section}.\arabic{num}.}\quad\itshape}
{\par\bigskip}

\newenvironment{prop}
{\refstepcounter{num}%
	\bigskip\noindent{\bf Proposition~\arabic{section}.\arabic{num}.}\quad\itshape}
{\par\bigskip}

\newenvironment{cor}
{\refstepcounter{num}%
	\bigskip\noindent{\bf Corollary~\arabic{section}.\arabic{num}.}\quad\itshape}
{\par\bigskip}

\newenvironment{lemma}
{\refstepcounter{num}%
	\bigskip\noindent{\bf Lemma~\arabic{section}.\arabic{num}.}\quad\itshape}
{\par\bigskip}

\newenvironment{example}
{\refstepcounter{num}%
	\bigskip\noindent{\bf Example~\arabic{section}.\arabic{num}.}\quad}
{\par\bigskip} 

\newenvironment{remark}
{\refstepcounter{num}%
	\bigskip\noindent{\bf Remark~\arabic{section}.\arabic{num}.}\quad}
{\par\bigskip}

\newenvironment{defin}
{\refstepcounter{num}%
	\bigskip\noindent{\bf Definition~\arabic{section}.\arabic{num}.}\quad}
{\par\bigskip}

\newenvironment{theo_with_name}[1]
{\refstepcounter{num}%
	\bigskip\noindent{\bf Theorem~\arabic{section}.\arabic{num}} (#1).\quad\itshape}
{\par\bigskip}

\newenvironment{lemma_with_name}[1]
{\refstepcounter{num}%
	\bigskip\noindent{\bf Lemma~\arabic{section}.\arabic{num}} (#1).\quad\itshape}
{\par\bigskip}

\renewcommand{\d}{\mathrm{d}}
\newcommand{\dL}{\d^{\Lambda}}

\DeclareMathOperator{\im}{Im}

\newcommand{\CC}{{\mathbb{C}}}   

\newcommand{\ZZ}{{\mathbb{Z}}}   
\newcommand{\QQ}{{\mathbb{Q}}}
\newcommand{\RR}{{\mathbb{R}}}

\newcommand\restr[2]{{
		\left.\kern-\nulldelimiterspace 
		#1 
		\littletaller 
		\right|_{#2} 
}}

\newcommand{\littletaller}{\mathchoice{\vphantom{\big|}}{}{}{}}

\newcommand{\mylabel}[1]{}

\begin{document}

	\title{The Cohomology of Solvmanifold SYZ Mirrors}

    \author[Cavenaghi]{Leonardo F. Cavenaghi}
\address{Institute of Mathematics and Informatics, Bulgarian Academy of Sciences}
\email{leonardofcavenaghi@gmail.com}

\author[Grama]{Lino Grama}
\address{Instituto de Matemática, Estatística e Computação Científica (IMECC) da Universidade Estadual de Campinas (Unicamp), Cidade Universitária, Campinas - SP, 13083-856.} 
\email{lino@ime.unicamp.br}

\author[Katzarkov]{Ludmil Katzarkov}
\address{College of Arts and Sciences, Department of Mathematics, University of Miami, \& Institute of the Mathematical Sciences of the Americas (IMSA). Ungar Bldg, 1365 Memorial Dr 515, Coral Gables, FL 33146. \&  International Center for Mathematical Sciences (ICMS), Sofia-Bulgaria. \& Institute of Mathematics and Informatics, Bulgarian Academy of Sciences} 
\email{lkatzarkov@gmail.com}

  \author[Martins]{Pedro Antonio Muniz Martins}
\address{Instituto de Matemática, Estatística e Computação Científica (IMECC) da Universidade Estadual de Campinas (Unicamp), Cidade Universitária, Campinas - SP, 13083-856.} 
\email{pedroa.muniz9@gmail.com}

\begin{abstract}
This paper investigates the geometric and cohomological properties of non-K\"ahler Strominger-Yau-Zaslow (herein non-K\"ahler SYZ) mirror symmetry for dual torus fibrations over solvmanifolds in the sense of Lau, Tseng and Yau. We are mainly concerned with three questions:

\textbf{(a)} How the Lau-Tseng-Yau notion of non-K\"ahler SYZ is related to the mapping of supersymmetric branes between symplectic and complex sides; \textbf{(b)} Finding explicit non-K\"ahler SYZ mirror pairs determined purely by Lie-theoretic data; \textbf{(c)} better understand the cohomological correspondence in the Lau-Tseng-Yau framework (given by a Fourier-Mukai transform), especially concerning the role of Tseng-Yau cohomology. 

We prove that the Fourier-Mukai transform introduced by Lau-Tseng-Yau exchanges type-A supersymmetric cycles, which are given by special Lagrangian sections equipped with flat $\mathrm{U}(1)$ connections, with type-B cycles, corresponding to line bundles whose connections satisfy the deformed Hermitian-Yang-Mills (dHYM) equation. 

We provide pure Lie-theoretic criteria for the existence of non-Kähler SYZ mirror pairs whose base manifolds are solvmanifolds. Applying these criteria, we construct new explicit families of mirror pairs from almost abelian and generalized Heisenberg Lie groups, and provide a complete classification of such pairs arising from nilpotent Lie groups. 

To contextualize the role of the Tseng-Yau cohomology, we link it to noncommutative geometry. We introduce the Tseng-Yau and Bott-Chern mirror bicomplexes. We show that (some of) their enclosed cohomologies reduce to the primitive Tseng-Yau and Bott-Chern cohomologies and that for basic forms they are isomorphic under the Fourier-Mukai transform.

As a last contribution, we discuss how to explicitly compute the Tseng-Yau and the Bott-Chern cohomology for the non-K\"ahler SYZ mirror pairs constructed here. Again, this is fully characterized by Lie-theoretic data.
\end{abstract}

\maketitle

\section{Introduction}

The Strominger-Yau-Zaslow (SYZ) conjecture \cite{Strominger_1996} proposes that mirror symmetry between two Calabi-Yau manifolds, $X$ and $\breve{X}$, can be understood geometrically through dual special Lagrangian torus fibrations, $\pi \colon X \to B$ and $\breve{\pi} \colon \breve{X} \to B$, over a common base $B$ \cite{Gross_2012, Gross_2008}. On the other hand, it is reasonable to expect \cite{ward2021homologicalmirrorsymmetryelliptic, Abouzaid2016, Abouzaid2013, Auroux_2007, Kapustin2009, Gross2017, Clarke2017, Cleyton2010} that the phenomenon of mirror symmetry is not restricted to Calabi-Yau or Fano manifolds. Indeed, in the semi-flat case (when singular fibers are absent), this duality is well-understood for certain classes of non-Kähler manifolds, including solvmanifolds and nilmanifolds \cite{Bedulli_2024, Grama_2023, Popovici2, Roychowdhury2024, Katzarkov2018, Minasian2016, hirolee, GARCIAFERNANDEZ20191059, Collins2023}.

\vspace{1em}

For dual torus fibrations of non-Kähler Calabi-Yau manifolds, Lau, Tseng, and Yau \cite{Lau_2015} construct a Fourier-Mukai-type transform acting on the sheaf of differential forms. They prove (Theorem \ref{mirror_Yau}) that this transform exchanges certain $\mathrm{SU}(n)$-structures on $X$ and $\breve{X}$, corresponding to what are termed \emph{type IIB} and \emph{type IIA} supersymmetric systems (see Section \ref{sec:SUN} for precise definitions). Furthermore, letting $H_{\d+\dL}(\breve X,\mathbb C)$ denote the Tseng-Yau cohomology of $\breve X$ and $H_{\partial+\overline\partial}(X,\mathbb C)$ denote the Bott-Chern cohomology of $X$, the Fourier-Mukai transform, when restricted to invariant (or basic) forms, induces an isomorphism (Theorem \ref{mirror_Yau_cohomologies}):
\begin{equation}
    H_{\text{bas},\d+\dL}^{n-p,q}(\breve X, \mathbb{C}) \cong H_{\text{bas},\partial+\overline\partial}^{p,q}(X, \mathbb{C}).
\end{equation}

In this context, Bedulli and Vannini \cite{Bedulli_2024} provide a construction of these non-K\"ahler supersymmetric mirror pairs by using left-invariant complete affine structures on solvable Lie groups. They prove that the T-dual pair arising from Arnold-Liouville Theorem \cite{Arnold_1989} (a construction which we recall very soon in this introduction) is a homogeneous space when the base is a solvmanifold $B = G/\Gamma$. Their proof also gives a complete characterization of them using the affine structure on $G$ (Theorem \ref{thm:dedulvan}).

\vspace{1em}

The present paper arises from the synthesis of three questions:
\begin{itemize}
    \item[\textbf{(a)}] To what extent can one establish Lie-theoretic criteria for the existence of supersymmetric mirror pairs for solvmanifolds?
    \item[\textbf{(b)}] What specific properties of Tseng-Yau cohomology enable it to play its central role in the framework of non-Kähler SYZ mirror symmetry?
    \item[\textbf{(c)}] How does the Lau-Tseng-Yau proposal fit the classical SYZ mapping of the A- and B-cycles between mirror partners?
\end{itemize}

To address these questions, we first establish the framework within which our results are developed (for more details, see Section \ref{sec:setup}).

\vspace{1em}

Let $B$ be a smooth manifold of dimension $n$ with an integral affine structure; that is, it carries an affine atlas whose transition functions take values in $\mathrm{GL}(n,\ZZ) \ltimes \RR^n$. The Arnold-Liouville Theorem \cite{Arnold_1989} says that this is equivalent to the existence of a Lagrangian torus fibration $\breve\pi:\breve{X} \to B$. Following the T-duality ansatz \cite{Strominger_1996,Leung2000,Bouwknegt_2004, Cavalcanti_2011}, one can build a dual torus bundle whose total space is:
\[
X := \left\{ (r,\nabla) \mid r \in B, \, \nabla \text{ is a flat } \mathrm{U}(1) \text{ connection over }\breve{\pi}^{-1}(r) \right\}.
\]
The bundle map $\pi \colon X \to B$ is obtained by forgetting the connection. The symplectic manifold $\breve X$ naturally admits a symplectic form $\breve\omega$, while the complex manifold $X$ is naturally equipped with a holomorphic volume form $\mathsf\Omega$. The components $(\breve{X}, \breve{\omega})$ and $(X, \mathsf{\Omega})$ form a \emph{T-dual pair}. When equipped with the complementary forms $\breve{\mathsf{\Omega}}$ and $\omega$ related to their endowed $\mathrm{SU}(n)$-structures, we say that the triples $(\breve{X}, \breve{\omega}, \breve{\mathsf{\Omega}})$ and $(X, \omega, \mathsf{\Omega})$ form a \emph{non-K\"ahler SYZ mirror pair} if the Fourier-Mukai transform introduced in \cite{Lau_2015} exchanges the type IIA supersymmetric $\mathrm{SU}(n)$-structure on $\breve X$ with the type IIB supersymmetric $\mathrm{SU}(n)$-structure on $X$.

\vspace{1em}

We answer Question \textbf{(a)} translating the characterization of type IIA and type IIB supersymmetric structures into our setup, where the base manifold $B=G/\Gamma$ is a solvmanifold. Here, $G$ is a simply connected Lie group with a complete left-invariant affine structure, and $\Gamma$ is a lattice. 

\vspace{1em}

Question \textbf{(c)} is answered more broadly by a straightforward, though highly non-obvious, adaptation of the arguments in \cite{Leung2000} to our setup. Although this adaptation was claimed to be possible in \cite{Lau_2015}, our contribution consists of precisely describing it.

\vspace{1em}

Question \textbf{(b)} is answered through a shift in focus. We begin with a question inspired by noncommutative geometry: is the Tseng-Yau cohomology isomorphic to the periodic cyclic homology of a dg-algebra? Although we do not fully resolve this, our partial answers shed light on the broader question.

\vspace{1em}

Next, we provide more details of these answers.

\vspace{1em}

\subsection{The answer to question (c)}

In \cite{Lau_2015}, it is claimed that the semi-flat A-branes in $\breve X$ are mirror to the semi-flat B-branes in $X$ under the Fourier-Mukai transform and that this should follow from the arguments in \cite{Leung2000}. However, this is not immediate, as in our broader setup, we lack a K\"ahler potential and a real Monge-Amp\`ere equation.

Despite these absences, we prove that if
\[
\xymatrix{
 (X, \omega, \mathsf{\Omega}) \ar[dr]_{\pi} & & (\breve{X}, \breve{\omega}, \breve{\mathsf{\Omega}}) \ar[dl]^{\breve{\pi}} \\
 & B &
}
\] 
forms a non-K\"ahler SYZ mirror pair, the Fourier-Mukai transform maps \emph{A-cycles} in $\breve{X}$ to \emph{B-cycles} in $X$. Roughly, A-cycles correspond to Lagrangian sections in the symplectic manifold $\breve{X}$, while B-cycles correspond to holomorphic line bundles on the complex mirror manifold $X$. Let us provide a better explanation.

As we detail in Section \ref{affine}, the $2$-form $\breve{\omega}$ is a natural symplectic form in $T^{\ast}B$. In local \emph{action-angle coordinates}, it can be written as $\breve{\omega} = \sum_{i=1}^{n} \d r_{i} \wedge \d \breve{\theta}_{i}$. First, we verify that a smooth generic section  $s \colon B \to \breve{X}$, locally given by $\breve{\theta}_i = s_i(r_1, \ldots, r_n)$, is Lagrangian if and only if its Jacobian matrix $H = \left(H_{kj} := \frac{\partial s_j}{\partial r_k}\right)$ is symmetric. To define the mirror analog of a Lagrangian, we note that any point in the fiber $\breve{\pi}^{-1}(b) \subset \breve{X}$, represented by the coordinates $\breve{\theta}$, defines a flat $\mathrm{U}(1)$-connection on its dual fiber $\pi^{-1}(b) \subset X$. The specification of section $s$ implies that the coordinates $\breve{\theta}_j$ acquire a base dependence $\breve{\theta}_j = s_j(r)$. Consequently, the pointwise family of flat connections along the fibers integrates into a global $\mathrm{U}(1)$ connection $\nabla_A$ on a line bundle over the total space of $X$:
\begin{equation}
 \nabla_A = \d + \mathbf{i} A = \d + \mathbf{i} \sum_{j=1}^n s_j(r) \d\theta_j.
\end{equation}
Proposition \ref{prop:integrability} shows that Lagrangians in $\breve X$ are in correspondence with holomorphic connections on the mirror side. Furthermore, Proposition \ref{prop:dHYM} ensures that we can translate the requirement for a Lagrangian $L$ to be of \emph{special phase}, i.e., 
\begin{equation}
\im\left(e^{-\mathbf{i}\vartheta} \restr{\breve{\mathsf{\Omega}}}{L}\right) = 0,
\end{equation} 
into the condition on the mirror that the connection satisfies the \emph{deformed Hermitian-Yang-Mills} (dHYM) equation. We prove:

\begin{Mtheorem}[{=Theorem \ref{thm:AndB-cyles}}]
 Let $(\breve{X}, \breve{\omega},\breve{\mathsf{\Omega}})$ and $(X, \omega, \mathsf{\Omega})$ be a non-K\"ahler SYZ mirror pair. Then the Fourier-Mukai transform maps a type-A supersymmetric cycle in $\breve{X}$ (given by a special Lagrangian section with a flat $\mathrm{U}(1)$ connection) to a type-B supersymmetric $2n$-cycle in $X$ (given by a line bundle whose connection satisfies $\mathrm{(dHYM)}$).
\end{Mtheorem}

\vspace{1em}

\subsection{The answer to question (a)}

To address Question (a), we seek pure Lie-theoretic conditions under which solvmanifolds admit supersymmetric mirror pairs. Let $G$ be a simply connected Lie group with a complete left-invariant integral affine structure and a lattice $\Gamma$. We set $B = G/\Gamma$. Per Theorem \ref{thm:dedulvan}, the dual total spaces, $X$ and $\breve{X}$, can be realized as homogeneous spaces of semi-direct products involving $G$ and the linear part $\rho$ of its affine representation. Although they naturally inherit invariant complex and symplectic structures, to check that they form a non-Kähler SYZ mirror pair, we must construct the canonical volume forms $\mathsf{\Omega}$ and $\breve{\mathsf{\Omega}}$ and ensure that they satisfy the right requirements of supersymmetric structures of types IIB and IIA, respectively, and verify that Fourier-Mukai $FT$ exchange them properly.

By explicitly constructing bases of left-invariant forms, we prove (Theorem \ref{Fouriermukaiomega}) that $FT(e^{2\omega}) = \breve{\mathsf{\Omega}}$. Moreover, we prove:

\begin{Mtheorem}[{=Theorem \ref{mirrorsolvmanifolds}}]
 Let $G$ be a simply connected Lie group endowed with a complete left-invariant affine structure, and let $\Gamma$ be a lattice such that the induced affine structure on the solvmanifold $G/\Gamma$ is integral. Denote by $\rho_* \colon \mathfrak{g} \to \mathfrak{gl}(n,\RR)$ the derivative of the linear part of the affine representation. Then, the dual torus bundles  $(X, \omega, \mathsf{\Omega})$ and $(\breve{X}, \breve{\omega}, \breve{\mathsf{\Omega}})$  form a non-Kähler SYZ mirror pair if and only if
 \[
 (\rho_{*}(X_i))_{i,i} = \frac{1}{2} \sum_{j=1}^n (\rho_{*}(X_j))_{i,j} 
 \]
 for all $1 \le i \le n$, where $\{X_1, \ldots, X_n\}$ is a basis of $\mathfrak{g}$.
\end{Mtheorem}

In Section \ref{SYZ solvmanifolds}, we apply it to construct explicit families of non-Kähler mirror pairs from almost abelian (Example \ref{almostabelianmirror}) Lie groups and the generalized Heisenberg group (Example \ref{ex:heisnberg}) and connect the existence of mirror pairs with the existence of geodesically complete flat Lorentzian metrics (Theorem \ref{thm:lorentzian}). 

Finally, by using what is known as Mal'cev's criterion \cite[Theorem 2.12]{Raghunathan1972}, we provide a complete classification of all non-Kähler SYZ mirror pairs arising from nilpotent Lie groups. To know
\begin{Mtheorem}[=Theorem \ref{mirror_nilmanifolds}]\label{mthm:nil}
	Let $G$ be a simply connected nilpotent Lie group with a complete left-invariant affine structure. Denote by $\rho$ the linear part of the affine representation and let $\{X_1,\ldots , X_n\}$ be a basis for $\mathfrak{g}$. Then, $G$ produces a non-Kähler SYZ mirror pair if and only if:
	\begin{itemize}
		\item[(1)] $C^{i}_{j,k} \in \QQ$ for all $1 \le i,j,k \le n$;
		\item[(2)] $(\rho_{*}(X_i))_{i,i} = \frac{1}{2}\sum_{j=1}^{n} (\rho_{*}(X_j))_{i,j}$ for all $1\le i \le n$;
		\item[(3)] There exists a lattice $\gamma \subset \mathfrak{g}$ such that $\exp(\rho_{*}(X))$ is conjugate to an integer matrix for all $X \in \gamma$.
	\end{itemize}
	Furthermore, the lattice in $G$ is explicitly given by $\Gamma := \exp(\gamma)$, making the induced affine structure on $G/ \Gamma$ integral and yielding the mirror pair:
	\[
	\xymatrix{
		(X, \mathsf{\Omega}, \omega) \ar[dr]_{\pi} & & (\breve{X}, \breve{\omega}, \breve{\mathsf{\Omega}}) \ar[dl]^{\breve{\pi}} \\
		& G/\Gamma &
	}
	\]
\end{Mtheorem}

We illustrate Theorem \ref{mthm:nil} in Example \ref{ex:plethora}.

\vspace{1em}

\subsection{The answer to question (b)}

On compact symplectic manifolds satisfying the $\d\dL$-lemma, Tseng-Yau cohomology degenerates to standard de Rham cohomology. Seeking further understanding of this behavior, we shift our focus toward noncommutative geometry. Inspired by Brylinski's definition of Poisson cohomology \cite{Brylinski1988} and the framework of periodic cyclic homology \cite{Kontsevich2008, NestTsyganCyclic}, we introduce a formal variable $u$ of degree $+2$ to construct the \emph{Tseng-Yau bicomplex} \[C_{\breve{X}}^{\bullet} = (\Omega^{\bullet}(\breve{X}, \CC)((u)), \d + u\dL, \d\dL)\] on the symplectic side, alongside a ``mirror'' complex \[\hat{C}_{X}^{\bullet} = (G^\bullet((u)), \partial + u\overline{\partial}, \partial\overline{\partial})\] on the complex side. Here, 
\[
G^{\bullet} \colon \quad \dots \xrightarrow{\partial} \Omega^{0, n}(X) \xrightarrow{\partial} \Omega^{0,n-1}(X) \oplus\Omega^{1,n}(X) \xrightarrow{\partial} \Omega^{0,n-2}(X) \oplus \Omega^{1,n-1}(X)\oplus \Omega^{2,n}(X) \xrightarrow{\partial} \dots
\]

We derive four distinct cohomologies from each bicomplex, which capture ``mixed-degree forms''. In particular:
\[
    H_{\d + \dL}^{k}(C^{\bullet}) := \frac{\ker(\d+u\dL \colon C^{k}\longrightarrow C^{k+1})}{\im(\d\dL \colon C^{k}\longrightarrow C^{k})},
\]
and
\[
    H_{\partial + \overline{\partial}}^{k}(\hat{C}^{\bullet}) := \frac{\ker(\partial +u\overline{\partial} \colon \hat{C}^{k}\longrightarrow \hat{C}^{k+1})}{\im(\partial\overline{\partial} \colon \hat{C}^{k}\longrightarrow \hat{C}^{k})},
\]
which we term (both) \emph{Bott-Chern cohomologies}, given their similarities with the usual Bott-Chern cohomology in complex geometry. 

We prove (Theorems \ref{BC_cohomology} and \ref{complex_BC}) that these new cohomologies contain (as graded vector spaces) the classical Tseng-Yau and Bott-Chern cohomologies as direct summands, but are generically infinite-dimensional. However, this infinite-dimensional ``noise'' vanishes when the bicomplexes are restricted to \emph{primitive forms} (i.e., those annihilated by the dual Lefschetz operator $\Lambda$). Theorem \ref{mthm:prim_BC} can be thought of as a `noncommutative geometry' definition for the primitive Tseng-Yau cohomology:

\begin{Mtheorem}[{=Theorem \ref{prim_BC}}]\label{mthm:prim_BC}
 Let $(\breve{X}, \breve{\omega})$ be a symplectic manifold of dimension $2n$. Restricting the Bott-Chern-type cohomology of the Tseng-Yau bicomplex to primitive forms recovers the classical primitive Tseng-Yau cohomology:
 \[
 PH^{k}_{\d+\dL}(C_{\breve{X}}^{\bullet}) \cong \begin{cases}
  \bigoplus_{m ~\mathrm{even}} PH_{\d+\dL}^{m}(\breve{X}, \CC), & \text{if } k \text{ is even,} \\[2ex]
  \bigoplus_{m ~\mathrm{odd}} PH_{\d+\dL}^{m}(\breve{X}, \CC), & \text{if } k \text{ is odd.}
 \end{cases}
 \]
 An analogous isomorphism holds for the primitive Bott-Chern cohomology on the complex mirror $X$.
\end{Mtheorem}

Lastly, we connect these introduced cohomologies with the mirror symmetry framework proposed by Lau, Tseng, and Yau:

\begin{Mtheorem}[{=Theorem \ref{Fourier_Mukai_sym}}]
 The Fourier-Mukai transform induces the following isomorphism for all $k$:
 \[
 H_{\text{bas},\partial + \overline{\partial}}^{k}(\hat{C}_{X}^{\bullet}) \cong H^{k}_{\text{bas},\d + \dL}(C_{\breve{X}}^{\bullet}).
 \]
\end{Mtheorem}

 \vspace{1em}

    In the last section of this paper, we explain how one can compute the Tseng-Yau cohomology for the non-Kähler mirror pairs constructed here. We do that with the help of a result due to Angella and Kasuya \cite{AngellaKasuya2019}, which reduces the computation on solvmanifolds to the Lie algebra cohomology of the associated semi-direct product. For almost abelian mirror pairs, we provide a complete, explicit description of both the Tseng-Yau cohomology and its mirror Bott-Chern counterpart. For the generalized Heisenberg group, we do not provide a closed formula due to combinatorial difficulties. However, we present all the ingredients to conclude the result.

    \vspace{1em }

    \subsection*{Acknowledgments}
    L.~F.~Cavenaghi and L.~Katzarkov are supported by the Simons Foundation, grant SFI-MPS-T-Institutes-00007697, and the Ministry of Education and Science of the Republic of Bulgaria, grant DO1-239/10.12.2024. L. Grama is partially supported by FAPESP grants no. 2021/04065-6, 2024/00923-6, and CNPq grant no.306021/2024-2. P.~A.~M.~Martins is supported by FAPESP grants no. 2024/07684-7, 2025/22312-1. 

    The authors thank Giovane Galindo, Daniel Pomerleano, and Alexander Vitanov for useful comments regarding the content and exposition.

    \vspace{1em}
    
	
	\section{Non-Kähler SYZ mirror construction}

    In this paper, we will explore two types of non-K\"ahler Calabi-Yau geometries. Following \cite{Lau_2015}, we will refer to them as type $\textrm{IIA}$ and $\textrm{IIB}$ supersymmetric systems (Definitions \ref{def:supersymmetrictypeiiB} and \ref{def:supersymmetrictypeiiA} below). These are high-dimensional analogs of the six-dimensional cases originating from string theory; cf. \cite{ Podest2018, Gen_cohom, Fino2015, Gray2012, Klaput2011} and the references therein.

\subsection{\texorpdfstring{$\mathrm{SU}(n)$}{SU(n)}-structure} 
\label{sec:SUN}

First, we need to recast some basic information on $\mathrm{SU}(n)$-structures.

\begin{defin}
 Let $X$ be a real manifold of dimension $2n$. An \emph{$\mathrm{SU}(n)$-structure} on $X$ is a pair $(\omega,\mathsf{\Omega})$ of differential forms satisfying the following:
 \begin{itemize}
  \item[(a)] $\mathsf{\Omega}$ is a nowhere-vanishing complex-valued decomposable $n$-form on $X$. It induces an almost complex structure $\mathsf{J}$, which determines the complexified tangent bundle decomposition 
        \[
        TX \otimes \CC = T^{(1,0)}X \oplus T^{(0,1)}X.
        \]
        Here, the subbundles are defined as
        \[
        T^{(0,1)}X := \{ v \in TX \otimes \CC \mid \iota_{v}\mathsf{\Omega} = 0\} \quad \text{and} \quad T^{(1,0)}X := \overline{T^{(0,1)}X}.
        \]
        By construction, $\mathsf{\Omega}$ is an $(n,0)$-form with respect to this almost complex structure $\mathsf{J}$.
  
  \item[(b)] Relative to this decomposition of the complexified tangent bundle, $\omega$ is a non-degenerate real $(1,1)$-form with respect to the almost complex structure $\mathsf{J}$. Furthermore, the bilinear form $\omega(\, \cdot \, , \mathsf{J} \, \cdot \, )$ defines an almost Hermitian metric on $X$.
 \end{itemize}
\end{defin}

Since both $\mathbf{i}^{-n} \mathsf{\Omega} \wedge \overline{\mathsf{\Omega}}$ and $\omega^{n}$ are real nowhere-vanishing top-forms, we have:
\[
\mathsf{\Omega} \wedge \overline{\mathsf{\Omega}} = \mathbf{i}^{n} \mathsf{F} \frac{\omega^{n}}{n!}
\]
for some nowhere-vanishing function $\mathsf{F} \colon X \to \RR$, which is called the \emph{conformal factor} of the $\mathrm{SU}(n)$-structure. 

\begin{defin}
 An $\mathrm{SU}(n)$-structure $(X,\omega, \mathsf{\Omega})$ is said to be \emph{semi-flat} if there exists a smooth manifold $B$ over which $X$ presents itself as the total space of a Lagrangian torus bundle $\mu \colon X \to B$, and $\omega, \mathsf{\Omega} \in \Omega^{*}_{B}(X, \CC)$. That is, both $\omega$ and $\mathsf{\Omega}$ are complex-valued \emph{basic}\footnote{Meaning they are constant along the fibers.} forms.
\end{defin}

\vspace{1em}

\subsection{Supersymmetric systems} 

The supersymmetric systems here considered are \emph{in essentia} $\mathrm{SU}(n)$-structures with some constraints.  

\begin{defin}\label{def:supersymmetrictypeiiB}
 An $\mathrm{SU}(n)$-structure $(X,\omega,\mathsf{\Omega})$ is \emph{supersymmetric of type $\textrm{IIB}$} if $\d\mathsf{\Omega} = 0$ and $\d(\omega^{n-1}) = 0$. A form $\omega$ satisfying this latter condition is said to be \emph{balanced}.
\end{defin}

To better grasp Definition \ref{def:supersymmetrictypeiiB}, let us spell out the conditions for $X^{2\times 3}$ (i.e., $n=3$). It is given by the following system of differential equations:
\begin{align}
 \d \mathsf{\Omega} &= 0 , \qquad \d (\omega^{2}) = 0, \\
 \mathsf{\Omega} \wedge \overline{\mathsf{\Omega}} &= -\mathbf{i} \mathsf{F} \cdot \frac{\omega^{3}}{3!}, \\
 \rho_{B} &= 2\mathbf{i}\, \partial \overline{\partial} (\mathsf{F}^{-1} \cdot \omega).
\end{align}
Here, $\rho_B$ is what is known in the literature as the \emph{Ramond-Ramond flux} (hereafter RR flux). It can be roughly thought of as an obstruction to $X$ being a Calabi-Yau manifold in the following sense: if $\d \omega = 0$ and $\mathsf{F}$ is constant, then $\rho_{B} = 0$ and $(X,\omega , \mathsf{\Omega})$ is Calabi-Yau.

\vspace{1em}

To define the high-dimensional type $\textrm{IIA}$ supersymmetric system, some preliminary definitions are necessary.

\begin{defin}
 Let $(X,\omega,\mathsf{\Omega})$ be an $\mathrm{SU}(n)$-structure. Recall that a real polarization $(\Lambda, U)$ with respect to $\omega$ is an integrable Lagrangian distribution $\Lambda \subset TU$ defined over a dense open subset $U \subset X$. We say that $(\Lambda, U)$ is \emph{special of phase $\vartheta \in \RR/2\pi\ZZ$} with respect to $\mathsf{\Omega}$ if $\mathsf{\Omega}|_{\Lambda_{x}} \in \RR_{>0} e^{\mathbf{i}\vartheta}$ for all $x \in U$.
 
 This polarization induces an orthogonal splitting of the tangent bundle $TU = \Lambda \oplus \Lambda^{\perp}$ (with respect to the induced Riemannian metric $\mathsf g=\omega(\cdot,J\cdot)$), which yields a corresponding decomposition of the space of $k$-forms:
 \[
    \Omega^{k}(U) = \bigoplus_{p+q = k} \Omega^{(p,q)^{\Lambda}}(U).
    \]
 Here, $\Omega^{(p,q)^{\Lambda}}(U)$ denotes the space of differential forms in $U$ that have $p$ components along $\Lambda$ and $q$ components along $\Lambda^{\perp}$. Associated with this decomposition is the natural projection $\pi_{\Lambda}^{p,q} \colon \Omega^{k}(U) \to \Omega^{(p,q)^{\Lambda}}(U)$.
\end{defin}

\begin{defin}\label{def:supersymmetrictypeiiA}
 Let $(X,\omega,\mathsf{\Omega})$ be an $\mathrm{SU}(n)$-structure with a special real polarization $(\Lambda, U)$. We say that the triple $(X,\omega,\mathsf{\Omega})$ is \emph{supersymmetric of type $\textrm{IIA}$} if:
 \begin{itemize}
  \item[(a)] $\d\omega = 0$;
  \item[(b)] $\d\left( \pi_{\Lambda}^{n,0} \circ \restr{\mathsf{\Omega}}{U} \right) = 0$;
  \item[(c)] $\d\left( \pi_{\Lambda}^{1,n-1} \circ \restr{\mathsf{\Omega}}{U} \right) = 0$.
 \end{itemize}
\end{defin}

Just as we did for Definition \ref{def:supersymmetrictypeiiB}, let us spell out Definition \ref{def:supersymmetrictypeiiA} for the case where $n=3$. We have:
\begin{align}
 \d \omega &= 0 , \qquad \d (\mathrm{Re}\,\mathsf{\Omega}) = 0, \\
 \mathsf{\Omega} \wedge \overline{\mathsf{\Omega}} &= -\mathbf{i} \mathsf{F} \cdot \frac{\omega^{3}}{3!}, \\
 \rho_{A} &= - \mathbf{i}\d \dL \left( \mathsf{F} \cdot \im\, \mathsf{\Omega} \right).
\end{align}
Once again, if $\d \mathsf{\Omega} = 0$ and $\mathsf{F}$ is constant, then $\rho_{A} = 0$ and $(X,\omega, \mathsf{\Omega})$ is Calabi-Yau.

\vspace{1em}

\subsection{Affine structures}\label{affine} 

Let $B$ be a smooth manifold of dimension $n$. Recall that an affine structure on $B$ is an atlas whose transition maps are affine. A given affine structure is termed \emph{integral} if the linear parts of the transition functions are in $\mathrm{GL}(n,\ZZ)$. By the Arnold-Liouville Theorem in classical mechanics~\cite{Arnold_1989}, an integral affine structure on a manifold $B$ is equivalent to a Lagrangian torus bundle $\breve{X} \to B$. This can be made explicit. We recall it here for the reader's convenience.

Let $q \in B$ and let $r_{1}, \ldots, r_{n}$ be local affine coordinates in a neighborhood $U$ of $q$. Notice that $\d r_{1} , \ldots, \d r_{n}$ globally generates $T^{*}U$. Consider the lattice $\Lambda^{*}:= \mathrm{span}_{\ZZ}\{\d r_{1}, \ldots, \d r_{n} \}$. Because the affine structure is integral, the quotient $\breve{X} := T^{*}B/\Lambda^{*}$ inherits the structure of a smooth manifold and defines a fibration over $B$. Furthermore, the fibers of $\breve{X} \to B$ define a Lagrangian distribution with respect to the canonical symplectic form induced on the quotient. In terms of coordinates, this symplectic form is written as $\breve{\omega} = \sum_{i=1}^n \d r_{i} \wedge \d \breve{\theta}_{i}$, where $\breve{\theta}_{1}, \ldots, \breve{\theta}_{n}$ are the fiber coordinates in $\breve{X}$.

In the spirit of T-duality \cite{Leung2000,Bouwknegt_2004, Auroux_2007, Cavalcanti_2011}, given a torus bundle $\breve{\pi} \colon \breve{X} \to B$, it is possible to build a \emph{dual} torus bundle whose total space is:
\[
X := \left\{ (r,\nabla) \mid r \in B, \, \nabla \text{ is a flat } \mathrm{U}(1) \text{ connection over }\breve{\pi}^{-1}(r) \right\}.
\]
The bundle map $\pi \colon X \to B$ is obtained by forgetting the connection. Let $p \in \breve{\pi}^{-1}(b) \subset \breve{X}$. The flat connection $\mathrm{U}(1)$ on the dual fiber $\pi^{-1}(b)$ can be explicitly described. Denoting the coordinates in $\breve{\pi}^{-1}(b) \subset \breve{X}$ by $\breve{\theta}_1,\ldots, \breve{\theta}_{n}$, and the coordinates in the dual torus $\pi^{-1}(b) \subset X$ by $\theta_1,\ldots, \theta_{n}$, we have:
\begin{equation}\label{equation:connection_formula}
 \nabla_{p} := \d + \mathbf{i}\sum_{j=1}^{n} \breve{\theta}_{j}(p) \d \theta_{j}
\end{equation}
where $\breve{\theta}_{j}(p)$ are constants evaluated at the point $p$, ensuring that $\nabla_{p}$ defines a flat $\mathrm{U}(1)$ connection in $\pi^{-1}(b) \subset X$.

Let $q \in B$ and let $r_{1} , \ldots, r_{n}$ be affine coordinates in a neighborhood $U$ of $q$. Similarly to the above, $\pi^{-1}(U) \simeq TU / \Lambda$, where $\Lambda = \mathrm{span}_{\ZZ}\left\{ \frac{\partial}{\partial r_{1}} , \ldots , \frac{\partial}{\partial r_{n}} \right\}$. Let $\theta_{1}, \ldots , \theta_{n} $ denote the fiber coordinates in $X$. The space $X$ admits a canonical complex structure defined by $z_{j} := \theta_{j} + \mathbf{i} r_{j}$. Consequently, $X$ inherits an holomorphic volume form given by $\mathsf{\Omega} := \d z_{1} \wedge \ldots \wedge \d z_{n}$.

\begin{defin}\label{def:T-dual}
 The pairs $(\breve{X}, \breve{\omega})$ and $(X, \mathsf{\Omega})$ are termed a \emph{T-dual pair}.
\end{defin}

In what follows, we shall denote by $\Omega^k_B(\breve{X},\CC)$ the sheaf of $k$-th \emph{basic} differential forms (i.e., complex-valued forms that are constant along the fibers) on $\breve{X}$. Similarly, $\Omega_B^k(X,\CC)$ is defined analogously for $X$. Notice that in the presence of an $\mathrm{SU}(n)$-structure that is supersymmetric of type IIA, we can decompose the basic forms according to their $(p,q)$-components. Hence, hereafter, we shall also use the notation $\Omega_B^{p,q}(\breve X)$.

\vspace{1em}

\subsection{Non-Kähler SYZ}
\label{sec:setup}
In this section, the analog of the SYZ mirror phenomenon for non-K\"ahler manifolds is presented following~\cite{Lau_2015, Bedulli_2024}. Let $(\breve{X}, \breve{\omega})$ and $(X,\mathsf{\Omega})$ be a T-dual pair.\footnote{The reader should be aware that we are not using the identical notation conventions as employed in \cite{Lau_2015}.} 

To formalize the mirror symmetry between these dual torus bundles, we define an operator that exchanges the base and fiber differentials. Let $\breve{\alpha} \in \Omega^{k}_{B}( \breve{X} , \CC )$. In local action-angle coordinates derived from the Arnold-Liouville theorem, this is written as:
\[
\breve{\alpha} = \sum_{|I|+|J| = k} f_{I,J}(r) \d\breve{\theta}_{J} \wedge \d r_{I}
\]
where $f_{I,J} \colon B \to \CC$ are smooth functions. Analogously, let $\alpha \in \Omega^{p,q}_{B}(X)$. In local coordinates,
\[
\alpha = \sum_{|I|=p \, , \, |J| = q } f_{I,J}(r) \d z_{I} \wedge \d\overline{z}_{J}
\]
where $f_{I,J} \colon B \to \CC$ are smooth functions, and $z_{j} = \theta_{j} + \mathbf{i} r_{j}$.

\begin{defin}
 The \emph{polarization switch} operator $\mathsf{P}$ is defined by its action on $\Omega^{p,q}_{B}(X)$, which is locally written as:
 \[
 \sum_{|I|=p \, , \, |J| = q } f_{I,J}(r) \d z_{I} \wedge \d\overline{z}_{J} \stackrel{\mathsf{P}}{\longmapsto} \sum_{|I|=p \, , \, |J|=q} f_{I,J}(r) \d \theta_{I} \wedge \d r_{J}.
 \]
\end{defin}

Consider the fiber product:
\[
\xymatrix{
 & X \times_{B} \breve{X} \ar[dl]_{p} \ar[dr]^{\breve{p}} & \\
 X \ar[dr]_{\pi} & & \breve{X} \ar[dl]^{\breve{\pi}} \\
 & B &
}
\]
where $p$ is the projection onto $X$ and $\breve{p}$ is the projection onto $\breve{X}$. The space $X \times_{B} \breve{X}$ admits a universal connection, locally written as \begin{equation}\label{eq:canonical-connection}\d + \sum_{j=1}^n \left(\mathbf{i} \theta_{j} \d \breve{\theta}_{j} - \mathbf{i} \breve{\theta}_{j} \d\theta_{j}\right).\end{equation} Its curvature is expressed as:
\[
\Theta = 2\mathbf{i} \sum_{j=1}^{n} \d\theta_{j} \wedge \d\breve{\theta}_{j}.
\]

\begin{defin}(Fourier-Mukai transform) 
    Let $\alpha \in \Omega^{\bullet,\bullet}_{B}(X)$. Its Fourier-Mukai transform is defined as:
 \[
 FT(\alpha) := \breve{p}_{*}\left( p^{*}(\mathsf{P}(\alpha)) \wedge \exp\left(\frac{\Theta}{2\mathbf i}\right) \right)
 \]
 where $\breve{p}_{*}$ denotes integration along the torus fibers.
 
 The Fourier-Mukai transform of $\breve{\alpha} \in \Omega^{*}_{B}( \breve{X},\CC )$ is conversely defined as:
 \[
 FT\left(\breve{\alpha} \right) := \mathsf{P}^{-1} \left( p_{*} \left(  \breve{p}^{*}(\breve{\alpha}) \wedge \exp\left(-\frac{\Theta}{2\mathbf i}\right) \right)  \right).
 \]
\end{defin}

In the sense of Lau-Tseng-Yau \cite{Lau_2015}, the SYZ mirror phenomenon for non-Kähler manifolds is stated as a symmetry between type IIA and type IIB supersymmetric systems via the Fourier-Mukai transform:

\begin{theo_with_name}{Lau, Tseng, Yau~\cite[Theorem 5.1]{Lau_2015}}\label{mirror_Yau} 
 Let $\omega$ be a torus-invariant real $(1,1)$-form on $X$ and let $\breve{\mathsf{\Omega}}$ be the Fourier-Mukai transform of $e^{2\omega}$.
 \begin{enumerate}
  \item $(X,\mathsf{\Omega}, \omega)$ defines an $\mathrm{SU}(n)$-structure on $X$ if and only if $(\breve{X},\breve{\omega},\breve{\mathsf{\Omega}})$ defines an $\mathrm{SU}(n)$-structure on $\breve X$. For such a pair, the conformal factor $\breve{\mathsf{F}}$ of $(\breve{X},\breve{\omega},\breve{\mathsf{\Omega}})$ is related to the conformal factor $\mathsf{F}$ of $(X,\mathsf{\Omega}, \omega)$ by $\mathsf{F} = 2^{2n}\breve{\mathsf{F}}^{-1}$.
  
  \item $(\breve{X},\breve{\omega},\breve{\mathsf{\Omega}})$ is supersymmetric of type IIA if and only if $(X,\mathsf{\Omega}, \omega)$ is supersymmetric of type IIB.
  
  \item Let $\breve{\rho}_{A}$ be the type IIA RR flux source current of $(\breve{X},\breve{\omega},\breve{\mathsf{\Omega}})$. Then its Fourier-Mukai transform is equal to the type IIB RR flux source current $\rho_{B}$ of $(X,\mathsf{\Omega}, \omega)$ up to a constant multiple.
 \end{enumerate}
\end{theo_with_name}

Theorem \ref{mirror_Yau} leads naturally to the following definition:

\begin{defin}
 Let $X \to B$ and $\breve{X}\to B$ be a T-dual pair. We call them a \emph{non-Kähler SYZ mirror pair} if $(X,\mathsf{\Omega}, \omega)$ is supersymmetric of type IIB and $(\breve{X},\breve{\omega},\breve{\mathsf{\Omega}})$ is supersymmetric of type IIA.
\end{defin}

\vspace{1em}

Let us next explain why the correspondence given by the Fourier-Mukai transform between $\mathrm{SU}(n)$ supersymmetric structures of type IIA and type IIB embodies mirror symmetry. In what follows, we shall prove that if
\[
\xymatrix{
 (X, \mathsf{\Omega}, \omega) \ar[dr]_{\pi} & & (\breve{X}, \breve{\omega}, \breve{\mathsf{\Omega}}) \ar[dl]^{\breve{\pi}} \\
 & B &
}
\] 
forms a non-K\"ahler SYZ pair, the Fourier-Mukai transform maps \emph{A-cycles} in $\breve{X}$ to \emph{B-cycles} in $X$. Differently from the Calabi-Yau context, the standard argumentation does not hold straightforwardly for general non-K\"ahler SYZ pairs due to the non-existence of a Kähler potential and a real Monge-Ampère equation. 

Let us start by clarifying our terminology. In this context, A-cycles correspond to Lagrangian sections in the symplectic manifold $\breve{X}$. B-cycles correspond to holomorphic line bundles (or their associated flat connections) supported on the complex mirror manifold $X$. The main result to be proven in this section is Theorem \ref{thm:AndB-cyles}. 

\vspace{1em}

From Section \ref{affine}, we have a natural symplectic form $\breve{\omega}$ on $T^{\ast}B$. In local action-angle coordinates, we can write $\breve{\omega} = \sum_{i=1}^{n} \d r_{i} \wedge \d \breve{\theta}_{i}$.

\begin{lemma}\label{lemma:lagrangian_symmetry}
 Let $s \colon B \to \breve{X}$ be a smooth generic section, locally given by $\breve{\theta}_i = s_i(r_1, \ldots, r_n)$. The section $s(B)$ is Lagrangian if and only if its Jacobian matrix $H = \left(H_{kj} := \frac{\partial s_j}{\partial r_k}\right)$ is symmetric.
\end{lemma}
\begin{proof}
This follows from the direct computation:
 \begin{align}
  s^*\breve{\omega} &= s^* \left( \sum_{i=1}^n \d r_i \wedge \d \breve{\theta}_i \right) = \sum_{i=1}^n \d r_i \wedge \d(s_i(r)) \\
  &= \sum_{i=1}^n \d r_i \wedge \left( \sum_{j=1}^n \frac{\partial s_i}{\partial r_j} \d r_j \right) = \sum_{1 \le i < j \le n} \left( \frac{\partial s_j}{\partial r_i} - \frac{\partial s_i}{\partial r_j} \right) \d r_i \wedge \d r_j.
 \end{align}
\end{proof}

To define what would be analogous to a Lagrangian in the \emph{mirror}, we must understand how this section transforms. As previously discussed via Equation~\eqref{equation:connection_formula}, any point in the fiber $\breve{\pi}^{-1}(b) \subset \breve{X}$, represented by the coordinates $\breve{\theta}$, defines a flat $\mathrm{U}(1)$-connection on its dual fiber $\pi^{-1}(b) \subset X$. By choosing the section $s \colon B \to \breve{X}$, the coordinates $\breve{\theta}_j$ acquire a base dependence $\breve{\theta}_j = s_j(r)$. Consequently, the pointwise family of flat connections along the fibers integrates into a global $\mathrm{U}(1)$ connection $\nabla_A$ on a line bundle over the total space of $X$:
\begin{equation}
 \nabla_A = \d + \mathbf{i} A = \d + \mathbf{i} \sum_{j=1}^n s_j(r) \d\theta_j.
\end{equation}
Proposition \ref{prop:integrability} below shows that Lagrangians are in direct correspondence with holomorphic connections.

\begin{prop}\label{prop:integrability}
 The connection $\nabla_A$ defines a holomorphic structure on the line bundle over $X$ (i.e., $F_A^{0,2} = 0$) if and only if the cycle $s(B) \subset \breve{X}$ is Lagrangian.
\end{prop}
\begin{proof}
 The curvature $2$-form of the connection is $F_A = \mathbf{i} \d A = \mathbf{i} \sum_{k,j} H_{kj} \d r_k \wedge \d\theta_j$. Using the canonical complex coordinates $z_j = \theta_j + \mathbf{i} r_j$ in $X$, we have $\d\theta_j = \frac{1}{2}(\d z_j + \d\overline{z}_j)$ and $\d r_k = \frac{1}{2\mathbf{i}}(\d z_k - \d\overline{z}_k)$. The $(0,2)$-component of the curvature form is given by:
 \[
 F_A^{0,2} = -\frac{1}{4} \sum_{j,k} H_{kj} \d\overline{z}_k \wedge \d\overline{z}_j = \frac{1}{8} \sum_{j,k} \left( H_{jk} - H_{kj} \right) \d\overline{z}_j \wedge \d\overline{z}_k.
 \]
 Thus, $F_A^{0,2} = 0 \iff H_{jk} = H_{kj}$. The proof is concluded by Lemma \ref{lemma:lagrangian_symmetry}.
\end{proof}

\vspace{1em}

In the context of supersymmetric of type IIA $\mathrm{SU}(n)$-structures, we can define the notion of \emph{special Lagrangian submanifolds} relying on the existing real polarization. Indeed, in our context the real polarization is given by a torus fibration. We shall say that a section $s:B\rightarrow \breve X$ is a \emph{special Lagrangian of phase $\vartheta$} if it is tangent to the distribution $\Lambda$. More precisely, $T_{s(x)}s(B)=\Lambda_{s(x)}$. In other words,

\begin{defin}\label{def:as-in-CY}
 Let $L\subset (\breve{X},\breve\omega,\breve{\mathsf\Omega})$ be a Lagrangian submanifold. We say that $L$ is a \emph{special Lagrangian of phase $\vartheta$} if
 \[
 \im\left(e^{-\mathbf{i}\vartheta} \restr{\breve{\mathsf{\Omega}}}{L}\right) = 0.
 \]
\end{defin}

Notice that Definition \ref{def:as-in-CY} is the usual one in the complex of mirror symmetry for Calabi-Yau manifolds. 

\vspace{1em}

Proposition \ref{prop:dHYM} below translates the requirement for a special Lagrangian to be of phase $\vartheta$ into a condition on the mirror. In the statement, $\mathsf g := \omega(\cdot, \mathsf{J}\cdot)$ is the induced Hermitian metric in $X$.

\begin{prop}\label{prop:dHYM}
 The cycle $s(B) \subset \breve{X}$ is a special Lagrangian of phase $\vartheta$ if and only if $\omega + \mathbf{i} F_A$ on $X$ satisfies the deformed Hermitian-Yang-Mills equation:
 \begin{equation}\label{eq:dHYM}
  \tag{dHYM}
  \im \det(\mathsf{g} + \mathbf{i} H) = (\tan \vartheta) \mathrm{Re} \det(\mathsf{g} + \mathbf{i} H).
 \end{equation}
\end{prop}

\begin{proof}
On the one hand,
 \[
    \omega + \mathbf{i} F_A = \sum_{j,k} (\mathsf{g}_{jk} + \mathbf{i} H_{jk}) \d r_j \wedge \d \theta_k,\quad  (\omega + \mathbf{i} F_A)^n = n! \det(\mathsf{g} + \mathbf{i} H) \d r_1 \wedge \d \theta_1 \wedge \ldots \wedge \d r_n \wedge \d \theta_n. 
    \]
Moreover,  
 \[
  \im\left( e^{-\mathbf{i}\vartheta} (\omega + \mathbf{i} F_A)^n \right) = 0\iff  \im\left( (\omega + \mathbf{i} F_A)^n \right) = (\tan \vartheta) \mathrm{Re}\left( (\omega + \mathbf{i} F_A)^n \right).
    \]
    We can rewrite the left hand side above in terms of $\mathsf g,~H$:
    \begin{equation}\label{eq:matrix_B}
 \im \det(\mathsf{g} + \mathbf{i} H) = (\tan \vartheta) \mathrm{Re} \det(\mathsf{g} + \mathbf{i} H)\iff  \im\left( e^{-\mathbf{i}\vartheta} (\omega + \mathbf{i} F_A)^n \right)=0.
    \end{equation}

    On the other hand, from the fact that $X$ and $\breve{X}$ are a non-Kähler SYZ mirror pair, one has that $FT(e^{\omega})=\breve{\mathsf\Omega}$. This implies that
    \[
    \breve{\mathsf \Omega} = \left( \d \breve{\theta}_{1} + \mathbf{i}\sum_{j} \mathsf{g}_{1j} \d r_j \right)\wedge \ldots \wedge \left(  \d \breve{\theta}_{n} + \mathbf{i}\sum_{j} \mathsf{g}_{nj} \d r_j \right).
    \]
     Because $\d \breve{\theta}_j|_{s(B)} = \sum_k H_{jk} \d r_k$, we have
 \[
    \restr{\breve{\mathsf{\Omega}}}{s(B)} = \det(H + \mathbf{i} \mathsf{g}) \d r_1 \wedge \ldots \wedge \d r_n. 
    \]
Spelling out, the cycle $s(B)$ is a special Lagrangian of phase $\vartheta$ if
 \[
    \im \det(H + \mathbf{i} \mathsf{g}) = (\tan \vartheta) \mathrm{Re} \det(H + \mathbf{i} \mathsf{g}). 
    \]
 Notice that $(H + \mathbf{i} \mathsf{g}) = \mathbf{i}(\mathsf{g} - \mathbf{i} H)$. Since $\mathsf{g}$ and $H$ are symmetric real matrices, $\det(\mathsf{g} - \mathbf{i} H)$ is the complex conjugate of $\det(\mathsf{g} + \mathbf{i} H)$. 
\end{proof}

 We are now in a position to state and prove the main result of this section.

\begin{theo}\label{thm:AndB-cyles}
 Let $(\breve{X}, \breve{\omega},\breve{\mathsf{\Omega}})$ and $(X, \mathsf{\Omega},\omega)$ be a non-K\"ahler SYZ mirror pair. Then the Fourier-Mukai transform maps a type-A supersymmetric cycle in $\breve{X}$ (given by a special Lagrangian section with a flat $\mathrm{U}(1)$ connection) to a type-B supersymmetric $2n$-cycle in $X$ (given by a line bundle whose connection satisfies $\mathrm{(dHYM)}$).
\end{theo}
\begin{proof}
 For a general type-A supersymmetric cycle in $\breve{X}$, we have a special Lagrangian section $s(B)$ together with a flat $\mathrm{U}(1)$ connection on it. A flat $\mathrm{U}(1)$ connection on the base can be written locally in the form $\d + \mathbf{i}\d e = \d + \mathbf{i} \sum_k \frac{\partial e}{\partial r_k} \d r_k$ for some smooth function $e = e(r)$. 
 
 By Propositions \ref{prop:integrability} and \ref{prop:dHYM}, the transform of the section $s(B)$ yields a connection $\nabla_A = \d + \mathbf{i} \sum s_j(r) \d\theta_j$ over $X$ that is integrable and satisfies the (dHYM) equation. In the fibered product, we have the connection (compare with Equation \eqref{eq:canonical-connection})
 \[
 \tilde{\nabla}_A = \d + \mathbf{i} \sum_{j=1}^n s_j(r) \d\theta_j + \mathbf{i} \sum_{k=1}^n \frac{\partial e}{\partial r_k} \d r_k.
 \]
 The newly added term $\mathbf{i} \d e$ is exact on the base. Thus, $F_{\tilde{A}} = F_A$. Consequently, $\tilde{\nabla}_A$ preserves both the integrability condition $F_{\tilde{A}}^{0,2} = 0$ and the (dHYM) equation.
\end{proof}
	
	\ 
	
	
	\section{Solvmanifolds}
	
\subsection{Affine structure on Lie groups} 
Let $B$ be a smooth manifold of dimension $n$ equipped with an affine structure. We can globalize the coordinate charts of $B$ by passing to its universal covering $\widetilde{B}$ and considering a \emph{developing map}, which is a local diffeomorphism $\operatorname{dev} \colon \widetilde{B} \to \RR^{n}$ (cf. \cite[Section 5.2]{Goldman_2022}). In the language of geometric structures, an affine structure on $B$ is an $(\operatorname{Aff}(\RR^{n}), \RR^{n})$-atlas. The development of this structure is entirely determined by a pair $(\operatorname{dev}, \alpha)$, where $\alpha \colon \pi_{1}(B) \to \operatorname{Aff}(\RR^{n})$ is the holonomy representation, such that for every deck transformation $\gamma \in \pi_{1}(B)$, the following diagram commutes:
\[
\xymatrix{
	\widetilde{B} \ar[r]^{\operatorname{dev}} \ar[d]_{\gamma} & \RR^{n} \ar[d]^{\alpha(\gamma)} \\
	\widetilde{B} \ar[r]_{\operatorname{dev}} & \RR^{n}
}
\]
Notice that the affine structure on $B$ is completely determined by its development. In this vocabulary, the affine structure is called \emph{integral} if $B$ admits a $(\mathrm{GL}(n,\ZZ) \ltimes \RR^{n},\RR^{n})$-atlas. Equivalently, this means that the image of the linear part of $\alpha$ is conjugate to a subgroup of $\mathrm{GL}(n,\ZZ)$.

\vspace{1em}

Let $G$ be a simply connected Lie group of dimension $n$. A left-invariant affine structure on $G$ is given by an affine atlas with the property that left-multiplication $L_{g} \colon G \to G$ acts as an affine automorphism for every $g \in G$. Because $G$ is simply connected, it coincides with its own universal cover. Consequently, the action of $L_{g}$ on $G$ induces a unique affine transformation $\alpha(g) \in \operatorname{Aff}(\RR^{n})$ in the developing space. This yields a commutative diagram:
\begin{equation}\label{eq:diagram-dev}
\xymatrix{
	G \ar[r]^{\operatorname{dev}} \ar[d]_{L_{g}} & \RR^{n} \ar[d]^{\alpha(g)} \\
	G \ar[r]_{\operatorname{dev}} & \RR^{n}
}
\end{equation}
Note that here $\alpha \colon G \to \operatorname{Aff}(\RR^{n})$ is a smooth affine representation of the Lie group $G$ itself, rather than a holonomy representation of a fundamental group.

The existence of such a structure restricts the algebraic nature of the Lie group:

\begin{lemma_with_name}{Auslander~\cite{Auslander}}\label{lem:auslander}
	Let $G$ be a simply connected Lie group of dimension $n$, equipped with a complete\footnote{i.e., geodesically complete. In terms of the developing map, it means it is a global diffeomorphism.} left-invariant affine structure and an affine representation $\alpha(g) = \operatorname{dev} \circ L_{g} \circ \operatorname{dev}^{-1}$. Then, $G$ is solvable.
\end{lemma_with_name}

\vspace{1em}

\subsection{SYZ for solvmanifolds}\label{SYZ solvmanifolds}
Let $G$ be as in the previous section, and consider a lattice $\Gamma \subset G$ such that the quotient manifold $B = G/\Gamma$ inherits an integral affine structure. Since $\pi_{1}(B) \cong \Gamma$, the condition that $B$ is integral is equivalent to requiring that $\im(\restr{\rho}{\Gamma})$ be conjugate to a subgroup of $\mathrm{GL}(n, \ZZ)$ within $\mathrm{GL}(n, \RR)$. Here, $\rho$ is the linear part of the affine presetation $\alpha$.

For $B$ taken as this, Bedulli and Vannini in \cite{Bedulli_2024} have proven that the T-dual pair $(X,\mathsf\Omega),~(\breve X,\breve{\omega})$ built under the machinery described in Section \ref{affine} are solvmanifolds, moreover, the following holds:

\begin{theo_with_name}{Bedulli and Vannini \cite{Bedulli_2024}}\label{thm:dedulvan}
	Let $\breve{X} \to B$ be the torus bundle induced by the integral affine structure on $B$, and let $X \to B$ be its dual. These spaces admit the following geometric characterizations:
	\begin{enumerate}
		\item The total space $\breve{X}$ is a homogeneous space under the transitive action of the semi-direct product $G \ltimes_{\rho^{\ast}} (\RR^{n})^{\ast}$, with stabilizer $\Gamma \ltimes_{\rho^{\ast}} (\ZZ^{n})^{\ast}$. Here, $\rho^{\ast}$ denotes the dual representation induced by $\rho$. Furthermore, the canonical symplectic form $\breve{\omega}$ on the cotangent bundle $T^{\ast}G \cong G \ltimes_{\rho^{\ast}} (\RR^{n})^{\ast}$ is left-invariant.
		
		\item Analogously, the total space $X$ is a homogeneous space under the transitive action of $G \ltimes_{\rho} \RR^{n}$, with stabilizer $\Gamma \ltimes_{\rho} \ZZ^{n}$. In this setting, the canonical integrable complex structure $\mathsf{\Omega}$ on the tangent bundle $TG \cong G \ltimes_{\rho} \RR^{n}$ is left-invariant.
	\end{enumerate}
\end{theo_with_name}

As a consequence of Theorem \ref{thm:dedulvan}, the canonical symplectic structure $\breve{\omega}$ on $T^{*}G \cong G \ltimes_{\rho^{*}}(\RR^{n})^{*}$ descends to the quotient $\breve{X}$. Similarly, the canonical integrable complex structure $\mathsf{\Omega}$ on $TG \cong G \ltimes_{\rho} \RR^{n}$ descends to the quotient $X$.

On the complex side $X$, let $\{e_{1}, \ldots, e_{2n}\}$ be a basis of left-invariant $1$-forms on $G \ltimes_{\rho} \RR^{n}$. We choose this basis such that at the identity they satisfy $\restr{e_{j}}{e} = \d\restr{\theta_{j}}{e}$ and $\restr{e_{j+n}}{e} = \d\restr{r_{j}}{e}$, where $\{r_{j}, \theta_{j}\}$ denote the action-angle coordinates in $X$. Recalling that the canonical complex coordinates are $z_j = \theta_j + \mathbf{i} r_j$, the holomorphic $n$-form in this frame is expressed as $\mathsf{\Omega} = \bigwedge_{j=1}^{n}(e_{j} + \mathbf{i}e_{j+n})$. Furthermore, we define a left-invariant $2$-form $\omega := \sum_{j=1}^{n} e_{j} \wedge e_{j+n}$ on $G \ltimes_{\rho} \RR^{n}$. This form also passes to the quotient, equipping $X$ with the triple $(X, \mathsf{\Omega}, \omega)$.

Applying the same procedure to $\breve{X}$, we obtain a basis of left-invariant $1$-forms $\{\breve{e}_{1}, \ldots , \breve{e}_{2n}\}$. The symplectic form is then given by $\breve{\omega} = \sum_{j=1}^{n} \breve{e}_{j} \wedge \breve{e}_{j+n}$. Analogously, we define the $(n,0)$-form $\breve{\mathsf{\Omega}} := \bigwedge_{j=1}^{n}(\breve{e}_{j} + \mathbf{i} \breve{e}_{j+n})$, resulting in the dual triple $(\breve{X}, \breve{\omega}, \breve{\mathsf{\Omega}})$.

The main goal of this section is to prove the following theorem, which provides an explicit algebraic criterion for these dual bundles to form a non-Kähler SYZ mirror pair. 

\vspace{1em}

\begin{theo}\label{mirrorsolvmanifolds}
	Let $G$ be a simply connected Lie group endowed with a complete left-invariant affine structure, and let $\Gamma$ be a lattice such that the induced affine structure on the solvmanifold $G/\Gamma$ is integral. Denote by $\rho$ the linear part of the affine representation. Then, the construction of the dual torus bundles presented in Section~\ref{affine},
	\[
	\xymatrix{
		(X, \mathsf{\Omega}, \omega) \ar[dr]_{\pi} & & (\breve{X}, \breve{\omega}, \breve{\mathsf{\Omega}}) \ar[dl]^{\breve{\pi}} \\
		& G/\Gamma &
	}
	\]
	forms a non-Kähler SYZ mirror pair if and only if
	\[
	(\rho_{*}(X_i))_{i,i} = \frac{1}{2} \sum_{j=1}^n (\rho_{*}(X_j))_{i,j} 
	\]
	for all $1 \le i \le n$.
\end{theo}

\vspace{1em}

The strategy for the proof is to verify the hypotheses of Theorem \ref{mirror_Yau}. This verification produces results of independent interest, which we establish in Theorems~\ref{Fouriermukaiomega} and \ref{balancedomega}. Later in this section, we present Examples \ref{almostabelianmirror}, \ref{ex:heisnberg} as a showcase for Theorem \ref{mirrorsolvmanifolds}.

\vspace{1em}

\begin{theo}\label{Fouriermukaiomega}
	Let $(X, \mathsf{\Omega}, \omega)$ and $(\breve{X}, \breve{\omega}, \breve{\mathsf{\Omega}})$ be the T-dual pair defined by the data $(G, \Gamma, \operatorname{dev})$. The Fourier-Mukai transform of the exponentiated symplectic form satisfies
	\[
    FT(e^{2\omega}) = \breve{\mathsf{\Omega}}.
    \]
\end{theo}
\begin{proof}
	We begin by establishing a basis of left-invariant $1$-forms on the semi-direct product $G \ltimes_{\rho} \RR^{n}$. For each $1 \le i \le n$ and any $(g,h) \in G \ltimes_{\rho} \RR^{n}$, we define:
	\[
	\restr{e_{i}}{(g,h)} := L^{*}_{(g,h)^{-1}} \left( \restr{\d\theta_i}{e} \right) \quad \text{and} \quad \restr{e_{i+n}}{(g,h)} := L^{*}_{(g,h)^{-1}} \left( \restr{\d r_i}{e} \right).
	\]
	Representing the pullback $L^{*}_{(g,h)^{-1}}$ in the coordinate basis $\{\d r_j, \d\theta_j\}$, we have:
	\[
	\restr{e_i}{(g,h)} = \sum_{j=1}^{n} \Phi_{i,j}(g) \restr{\d\theta_j}{(g,h)}, \quad  \restr{e_{i+n}}{(g,h)} = \sum_{j=1}^{n} \Phi_{i,j}(g)  \restr{\d r_{j}}{(g,h)}.
	\]
	Analogously, for the dual side $G \ltimes_{\rho^{*}} \RR^{n}$, the corresponding basis $\{\breve{e}_i, \breve{e}_{i+n}\}$ satisfies:
	\[
	\restr{\breve{e}_i}{(g,\tilde{h})} = \sum_{j=1}^{n} \Phi_{i,j}(g)  \restr{\d r_j}{(g,\tilde{h})}, \quad \restr{\breve{e}_{i+n}}{(g,\tilde{h})} = \sum_{j=1}^{n} (\Phi^{-1})_{j,i}(g) \restr{\d\breve{\theta}_{j}}{(g, \tilde{h})},
	\]
	for $g \in G$ and $(g,\tilde{h}) \in G \ltimes_{\rho^{*}} \RR^{n}$. 
	
	We can rewrite the symplectic form $\omega = \sum e_i \wedge e_{i+n}$ on the base as:
	\[
	\omega = \sum^{n}_{j,k=1} (\Phi^{t} \Phi)_{k,j} \d\theta_j \wedge \d r_{k}.
	\]
	Similarly, the holomorphic volume form on the dual side can be expanded as:
	\[
	\breve{\mathsf{\Omega}} = \bigwedge_{i=1}^{n}(\breve{e}_{i} + \mathbf{i} \breve{e}_{i+n}) = \bigwedge_{i=1}^{n} \left( [\Phi \d r]_{i} + \mathbf{i} [(\Phi^{-1})^{t} \d\breve{\theta}]_{i} \right).
	\]
	Factoring out the determinant of the transformation, we obtain:
	\[
	\breve{\mathsf{\Omega}} = \det(\Phi)^{-1} \bigwedge_{i=1}^{n} \left( \d\breve{\theta}_{i} - \mathbf{i} (\Phi^{t} \Phi \d r)_{i} \right).
	\]
	Applying the Fourier-Mukai transform to $e^{2\omega}$ involves a fiberwise integration that identifies the $\d\theta$ components with the dual $\d\breve{\theta}$ coordinates. The resulting transformed form is:
	\[
	FT(e^{2\omega}) = \bigwedge_{i=1}^{n} \left( \d\breve{\theta}_i - \mathbf{i} (\Phi^{t} \Phi \d r)_{i} \right).
	\]
	Comparing the last two expressions, we conclude:
	\[
	FT(e^{2\omega}) = \det(\Phi) \cdot \breve{\mathsf{\Omega}}.
	\]
	We finish the proof by showing that $\det(\Phi) \equiv 1$.

    \vspace{1em}
	
	Fix $g \in G$ and let $L_{g}$ be left multiplication in the affine coordinates. Set $\tilde{L}_{g}$ to be the same operator but written in standard coordinates. From the commutative diagram \eqref{eq:diagram-dev}, it is straightforward that:
	\[
	\alpha(g) \circ \operatorname{dev} = \operatorname{dev} \circ \tilde{L}_{g}.
	\]
	This defines the affine representation of $G$. Differentiating this yields $\d L_{g} = \rho(g)$, where $\rho$ is the linear part of the affine representation.
	
	Let $g = \exp(X_1)\cdot \ldots \cdot \exp(X_m) \in G$, where $X_1,\ldots, X_m \in \mathfrak{g}$. Then:
	\[
	\rho(g) = \rho(\exp(X_1)) \cdot \ldots \cdot \rho(\exp(X_m)) = \exp(\rho_{*}(X_1)) \cdot \ldots \cdot \exp(\rho_{*}(X_m)),
	\]
	from which it follows that:
	\[
	\rho_{*}(X_i) = \d \rho_{e}(X_i) = \restr{\frac{\d}{\d t}}{t=0} \rho(\exp(tX_i))  = \restr{\frac{\d}{\d t}}{t=0} \d L_{\exp(tX_i)} = \nabla_{X_i}.
	\]
	Here, $\nabla$ represents the affine connection, meaning it is a flat, torsion-free affine connection.
	
	According to \cite{Segal1991}, because $G$ has a complete left-invariant affine structure, the operator
	\begin{align}
		\mathfrak{R}_{X} \colon \mathfrak{g} &\longrightarrow \mathfrak{g} \\
		Y & \longmapsto \nabla_{Y}X
	\end{align}
	is a nilpotent endomorphism for all $X \in \mathfrak{g}$. Thus, its eigenvalues are identically zero, implying $\mathrm{tr}(\mathfrak{R}_{X}) = 0$ for all $X \in \mathfrak{g}$. 
	
	Because the existence of a lattice implies that $G$ is unimodular, the trace of the adjoint representation vanishes: $\mathrm{tr}(\mathrm{ad}_{X}) = \mathrm{tr}(Y \mapsto [X,Y]) = 0$ for all $X \in \mathfrak{g}$. Combining this with the torsion-free symmetry of the affine connection ($\nabla_X Y - \nabla_Y X = [X,Y]$), we obtain:
	\[
	0 = \mathrm{tr}(Y \mapsto [X,Y]) = \mathrm{tr}(Y \mapsto \nabla_{X}Y) - \underbrace{\mathrm{tr}(Y \mapsto \nabla_{Y}X)}_{\mathrm{tr}(\mathfrak{R}_X)=0}.
	\]
	Thus, $\mathrm{tr}(Y \mapsto \nabla_{X} Y) = 0$ for all $X \in \mathfrak{g}$. Consequently, for every $g \in G$, we have:
	\[
	\det(\rho(g)) =  \exp(\mathrm{tr}(\nabla_{X_1})) \cdot \ldots \cdot \exp(\mathrm{tr}(\nabla_{X_m})) = 1 \implies \det(\Phi(g)) = 1. 
	\]
\end{proof}

\vspace{1em}

Theorem \ref{balancedomega} below furnishes a characterization for the $2$-form $\omega$ to be balanced purely in terms of Lie-theoretic data and the linear part $\rho$ of the affine representation. Before we do so, let us prove the following algebraic lemma. While likely well-known, we include it for completeness and to fix our notation.

\begin{lemma}\label{lemmaliealgebra}
	Let $G \ltimes_{\rho} \RR^{n}$ be a semi-direct product Lie group, where $\dim(G) = n$. Denote by $\mathfrak{g}$ the Lie algebra of $G$ and fix a basis $X_1,\dots,X_n$ of $\mathfrak{g}$ with structure constants $C^{k}_{i,j}$, i.e., $[X_i,X_j] = \sum_{k=1}^n C^{k}_{i,j} X_k$. Let $Y_1,\dots,Y_n$ be the standard basis of $\RR^n$. Then the Lie algebra of $G \ltimes_\rho \RR^n$ is $\mathfrak{h} = \mathfrak{g} \ltimes_{\rho_{\ast}} \RR^n$, with basis $\{S_1,\dots,S_{2n}\} := \{X_1,\dots,X_n,Y_1,\dots,Y_n\}$. If $c^{p}_{q,r}$ (with $1\le p,q,r\le 2n$) denotes the structure constants of $\mathfrak{h}$ with respect to this basis (so that $[S_q,S_r] = \sum_p c^{p}_{q,r} S_p$), then:
	\[
	c^{p}_{q,r} =
	\begin{cases}
		C^{p}_{q,r}, & 1\le p,q,r\le n,\\[4pt]
		(\rho_{\ast}(X_q))_{p-n,\,r-n}, & 1\le q\le n,\; n+1\le p,r\le 2n,\\[4pt]
		-(\rho_{\ast}(X_r))_{p-n,\,q-n}, & n+1\le q\le 2n,\; 1\le r\le n,\; n+1\le p\le 2n,\\[4pt]
		0, & \text{otherwise}.
	\end{cases}
	\]
\end{lemma}
\begin{proof}
	The Lie bracket in $\mathfrak{h} = \mathfrak{g} \ltimes_{\rho_{\ast}} \RR^n$ is given by:
	\[
	[(X,Y),\,(X',Y')] = \bigl([X,X']_{\mathfrak{g}},\; \rho_{\ast}(X)Y' - \rho_{\ast}(X')Y \bigr),
	\]
	for $X,X'\in\mathfrak{g}$ and $Y,Y'\in\RR^n$. In the combined basis, we compute:
	\begin{itemize}
		\item For $1\le i,j,k\le n$: $[X_j,X_k] = \sum_{i=1}^n C^{i}_{j,k} X_i$, and hence $c^{i}_{j,k} = C^{i}_{j,k}$ and $c^{p}_{j,k} = 0$ for $p>n$.
		\item For $1\le j\le n$ and $1\le a\le n$ (so $q=j$, $r=n+a$):
		\[
		[X_j, Y_a] = \rho_{\ast}(X_j)Y_a = \sum_{b=1}^n (\rho_{\ast}(X_j))_{b,a} Y_b.
		\]
		Thus, $c^{\,n+b}_{\,j,\,n+a} = (\rho_{\ast}(X_j))_{b,a}$ for $1\le b\le n$.
		\item By antisymmetry, for $1\le a\le n$ and $1\le j\le n$:
		$[Y_a, X_j] = -[X_j,Y_a] = -\sum_{b=1}^n (\rho_{\ast}(X_j))_{b,a} Y_b$, yielding $c^{\,n+b}_{\,n+a,\,j} = -(\rho_{\ast}(X_j))_{b,a}$.
		\item For $1\le a,b\le n$: $[Y_a,Y_b]=0$, meaning all structure constants with both lower indices in the translation part identically vanish. \qedhere
	\end{itemize}
\end{proof}

\vspace{1em}

\begin{theo}\label{balancedomega}
	Consider a semi-direct product Lie group $G \ltimes_{\rho} \RR^{n}$ with $\dim(G)=n$, and let $\{e_1,\dots,e_{2n}\}$ be a basis of left-invariant $1$-forms dual to the basis $\{X_1,\dots,X_n,Y_1,\dots,Y_n\}$ of its Lie algebra. Define the $2$-form:
	\[
	\omega = \sum_{k=1}^{n} e_k \wedge e_{k+n}.
	\]
	Let $C^{k}_{i,j}$ be the structure constants of $\mathfrak{g}$, and let $\rho_{\ast} \colon \mathfrak{g} \to \mathfrak{gl}(n,\RR)$ be the derivative of the representation $\rho$. Then $\omega$ is balanced (i.e., $\d(\omega^{n-1})=0$) if and only if for every $i=1,\dots,n$:
	\[
	\sum_{\substack{j=1 \\ j\neq i}}^n \Bigl( (\rho_{\ast}(X_j))_{i,j} - (\rho_{\ast}(X_i))_{j,j} - C^{j}_{i,j} \Bigr) = 0.
	\]
\end{theo}
\begin{proof}
	The Maurer-Cartan equations give $\d e_p = -\frac12\sum_{q,r} c^{p}_{q,r} e_q\wedge e_r$. Using the explicit constants derived in Lemma~\ref{lemmaliealgebra}, we obtain:
	\begin{align}
	    \d e_i &= -\sum_{1\le j<k\le n} C^{i}_{j,k}\, e_j\wedge e_k, &(1\le i\le n), \\
	    \d e_{i+n} &= -\sum_{j,k=1}^{n} (\rho_{\ast}(X_j))_{i,k}\, e_j\wedge e_{k+n}, &(1\le i\le n).
	\end{align}
	
	Since $\omega = \sum_{k=1}^n e_k\wedge e_{k+n}$, its exterior derivative expands as:
	\[
	\d\omega = \sum_{k=1}^n \bigl( \d e_k\wedge e_{k+n} - e_k\wedge \d e_{k+n} \bigr).
	\]
	Substituting the differentials:
	\[
	\d\omega = -\sum_{k=1}^n\sum_{1\le i<j\le n} C^{k}_{i,j}\, e_i\wedge e_j\wedge e_{k+n} + \sum_{k=1}^n\sum_{i,j=1}^n (\rho_{\ast}(X_j))_{i,k}\, e_i\wedge e_j\wedge e_{k+n}.
	\]
	In the second sum, we utilize antisymmetry to rewrite $\sum_{i,j} (\rho_{\ast}(X_j))_{i,k} e_i\wedge e_j = \sum_{1\le i<j\le n} \bigl( (\rho_{\ast}(X_j))_{i,k} - (\rho_{\ast}(X_i))_{j,k} \bigr) e_i\wedge e_j$. Therefore:
	\begin{equation}\label{eq:domega}
		\d\omega = \sum_{k=1}^n\sum_{1\le i<j\le n} \Bigl( (\rho_{\ast}(X_j))_{i,k} - (\rho_{\ast}(X_i))_{j,k} - C^{k}_{i,j} \Bigr) e_i\wedge e_j\wedge e_{k+n}. 
	\end{equation}
	
	Next, $\omega^{n-2}$ is given by the standard expansion of a power of a sum of $n$ commuting terms:
	\begin{equation}\label{eq:omegan-2}
		\omega^{n-2} = (n-2)! \sum_{1\le i_1<\cdots<i_{n-2}\le n} e_{i_1}\wedge e_{i_1+n}\wedge\cdots\wedge e_{i_{n-2}}\wedge e_{i_{n-2}+n}. 
	\end{equation}
	
	Now, $\d(\omega^{n-1}) = (n-1)\,\d\omega\wedge\omega^{n-2}$. Wedging Equations \eqref{eq:domega} and \eqref{eq:omegan-2} yields a $(2n-1)$-form. For a fixed $r\in\{1,\dots,n\}$, consider the standard basis $(2n-1)$-form with the $r$-th translation differential omitted:
	\[
	\Phi_r = e_1\wedge e_{n+1}\wedge \cdots\wedge e_r\wedge \widehat{e_{n+r}}\wedge\cdots\wedge e_n\wedge e_{2n}.
	\]
	In the wedge product, a non-zero contribution to $\Phi_r$ occurs precisely when we select a term from $\d\omega$ with indices $\{i,j,k\}$ such that $\{i,j,k\} \cup \{i_1,\dots,i_{n-2}\} = \{1,\dots,n\}$ and all $n-1$ translation indices are distinct. This mathematically forces $k$ to equal either $i$ or $j$, and the set $\{i_1,\dots,i_{n-2}\}$ must be $\{1,\dots,n\}\setminus\{r,s\}$ where $s$ is the index among $i,j$ that matches $k$. A careful parity analysis of the permutation sign shows that the coefficient of $\Phi_r$ in $(n-1)\d\omega\wedge\omega^{n-2}$ is, up to an overall non-zero constant:
	\[
	\sum_{\substack{s=1\\ s\neq r}}^n \Bigl( (\rho_{\ast}(X_s))_{r,s} - (\rho_{\ast}(X_r))_{s,s} - C^{s}_{r,s} \Bigr).
	\]
	Hence:
	\[
	\frac{1}{(n-1)(n-2)!}\,\d(\omega^{n-1}) = \sum_{r=1}^n \left( \sum_{s\neq r} \bigl( (\rho_{\ast}(X_s))_{r,s} - (\rho_{\ast}(X_r))_{s,s} - C^{s}_{r,s} \bigr) \right) \Phi_r.
	\]
	Thus, $\d(\omega^{n-1})=0$ if and only if for every $r=1,\dots,n$:
	\[
	\sum_{s\neq r} \Bigl( (\rho_{\ast}(X_s))_{r,s} - (\rho_{\ast}(X_r))_{s,s} - C^{s}_{r,s} \Bigr) = 0. 
	\]
	Renaming the dummy indices $r=i$ and $s=j$ concludes the proof.
\end{proof}

\vspace{1em}

\begin{cor}
	Let $G$ be a Lie group with a complete left-invariant affine structure and a lattice $\Gamma$ that makes the induced affine structure on $G/\Gamma$ integral. Denote by $\rho$ the linear part of the affine representation. Consider the semi-direct product Lie group $G \ltimes_{\rho} \RR^{n}$ with $\dim(G)=n$, and let $\{e_1,\dots,e_{2n}\}$ be the standard basis of left-invariant $1$-forms. The $2$-form
	\[
	\omega = \sum_{k=1}^{n} e_k \wedge e_{k+n}
	\]
	is balanced if and only if
	\[
	(\rho_{*}(X_i))_{i,i} = \frac{1}{2} \sum_{j=1}^n (\rho_{*}(X_j))_{i,j} \quad \text{for all } 1 \le i \le n.
	\]
	Equivalently, if $\nabla$ denotes the affine connection described by the Christoffel symbols $\Gamma_{j,k}^i$, then
	\[
	\Gamma^{i}_{i,i} = \frac{1}{2} \sum_{j=1}^n \Gamma^{i}_{j,j} \quad \text{for all } 1 \le i \le n.
	\]
\end{cor}
\begin{proof}
	Because $G$ admits a complete left-invariant affine structure with a compatible lattice, it must be unimodular. Thus, $\mathrm{tr}(\mathrm{ad}_{X}) = 0$ for all $X \in \mathfrak{g}$. In terms of structure constants, this means:
	\[
	\sum_{j=1}^{n} C^{j}_{i,j} = 0 \quad \text{for all } 1 \le i \le n.
	\]
	Substituting this directly into the result of Theorem~\ref{balancedomega} yields:
	\begin{equation}\label{eq:rho-condition-balanced}
		\sum_{\substack{j=1 \\ j \neq i} }^n \Bigl( (\rho_{\ast}(X_j))_{i,j} - (\rho_{\ast}(X_i))_{j,j} \Bigr) = 0 \quad \text{for all } 1 \le i \le n.
	\end{equation}
	
	Denote by $\nabla$ the affine connection given by the left-invariant affine structure on $G$, and let $\Gamma^{i}_{j,k}$ be its corresponding Christoffel symbols. Following the proof of Theorem \ref{Fouriermukaiomega}, we established that $\rho_{*}(X_i) = \nabla_{X_i}$. Therefore, Equation \eqref{eq:rho-condition-balanced} can be rewritten explicitly in terms of the Christoffel symbols as:
	\begin{equation}\label{eq:Gamma}
		\sum_{\substack{j=1 \\ j \neq i} }^n \left( \Gamma^{i}_{j,j} - \Gamma^{j}_{i,j} \right) = 0 \quad \text{for all } 1 \le i \le n.
	\end{equation}
	
	Recall from the proof of Theorem \ref{Fouriermukaiomega} that the trace of the connection vanishes identically: $\mathrm{tr}(Y \mapsto \nabla_X Y) = 0$ for all $X \in \mathfrak{g}$. Consequently, $\sum_{j=1}^{n} \Gamma^{j}_{i,j} = 0$ for all $1 \le i \le n$. We can therefore split the second summation in Equation \eqref{eq:Gamma} as:
	\[
	\sum_{\substack{j=1 \\ j \neq i}}^n \Gamma^{j}_{i,j} = \left( \sum_{j=1}^{n} \Gamma^{j}_{i,j} \right) - \Gamma^{i}_{i,i} = 0 - \Gamma^{i}_{i,i} = -\Gamma^{i}_{i,i}.
	\]
	Substituting this back into Equation \eqref{eq:Gamma} leaves:
	\[
	\sum_{\substack{j=1 \\ j \neq i}}^n \Gamma^{i}_{j,j} + \Gamma^{i}_{i,i} = 0.
	\]
	By adding $\Gamma^{i}_{i,i}$ to both sides, we complete the sum over all $j$ on the left side, resulting in:
	\[
	\sum_{j=1}^n \Gamma^{i}_{j,j} = 2\Gamma^{i}_{i,i} \implies \Gamma^{i}_{i,i} = \frac{1}{2} \sum_{j=1}^{n} \Gamma^{i}_{j,j}.
	\]
	Returning this expression to the $\rho_{\ast}$ notation completes the proof.
\end{proof}
	
Let us present two examples verifying Theorem \ref{mirrorsolvmanifolds}.

\begin{example}\label{almostabelianmirror}
	Let $G = \RR \ltimes_{e^{x_{1}A}} \RR^{n-1}$ be an almost abelian Lie group endowed with a left-invariant complete affine structure given by the developing map $\operatorname{dev} = \mathrm{id}$. The affine representation is explicitly given by:
	\[
	\alpha(x_1, \ldots, x_n) v = \operatorname{dev} \circ L_{(x_1, \ldots, x_n)} \circ \operatorname{dev}^{-1}(v) = \begin{bmatrix}
		1  & 0 \\
		0  &  e^{x_1 A}
	\end{bmatrix}v +  \begin{bmatrix}
		x_1  \\
		\vdots  \\
		x_n
	\end{bmatrix}
	\]
	for $(x_1 , \ldots , x_n) \in G $ and every $v \in \RR^{n}$. The linear part is $\rho \colon G \to \mathrm{GL}(n,\RR)$, where:
	\[
	\rho(x_1 , \ldots, x_n) = \begin{bmatrix}
		1  & 0 \\
		0  &  e^{x_1 A}
	\end{bmatrix}.
	\]
	In Appendix \ref{ap:lattice}, we prove that there exists a lattice $\Gamma$ such that the induced affine structure on $B = G/\Gamma$ is integral. Assume $\mathrm{tr}(A) = 0$.
	
	Let $\{X_1 , \ldots, X_n\}$ be a basis of $\mathfrak{g}$. From the group law in $G$, we compute the derivative at the identity $e_G$:
	\[
	\rho_{*}(X_1) = \d \rho_{e_{G}}(X_1) = \begin{bmatrix}
		0  & 0 \\
		0 & A
	\end{bmatrix}
	\]
	and $\rho_{*}(X_i) = 0$ for all $i \neq 1$. Consequently, the diagonal entries vanish:
	\[
	(\rho_{*}(X_i))_{i,i} = 0 \quad \text{for all } 1 \le i \le n,
	\]
	and the column sums also vanish:
	\[
	\sum_{j=1}^{n} (\rho_{*}(X_j))_{i,j} = (\rho_{*}(X_1))_{i,1} = 0.
	\]
	By Theorem~\ref{mirrorsolvmanifolds}, the dual torus bundles $X$ and $\breve{X}$ define a non-Kähler SYZ mirror pair $(X, \mathsf{\Omega}, \omega)$ and $(\breve{X}, \breve{\omega}, \breve{\mathsf{\Omega}})$.
\end{example}

\vspace{1em}

Example \ref{ex:heisnberg} below is a generalization of the example presented in Section 7.1 in \cite{Lau_2015} and provides a simple presentation of the example presented in Section 7.2 in the same reference.

\begin{example}\label{ex:heisnberg}
	Let $\mathcal{H}_{2n+1}(\RR)$ be the generalized Heisenberg group, whose elements are represented by matrices of the form:
	\[
	\begin{bmatrix}
		1 & \overline{a} & c \\
		0  &  I_{n}  &  \overline{b} \\
		0  & 0 & 1
	\end{bmatrix}
	\]
	where $\overline{a}$ is a real $(1\times n)$ row vector, $\overline{b}$ is a real $(n \times 1)$ column vector, and $c \in \RR$. It admits a left-invariant affine structure given by the developing map:
	\begin{align}
		D \colon \mathcal{H}_{2n+1}(\RR) & \longrightarrow \RR^{2n+1} \\
		\begin{bmatrix}
			1 & \overline{a} & c \\
			0  &  I_{n}  &  \overline{b} \\
			0  & 0 & 1
		\end{bmatrix} & \longmapsto (\overline{a}, c , \overline{b}).
	\end{align}
	Furthermore, $\mathcal{H}_{2n+1}(\RR)$ contains a lattice $\Gamma$ given by integer matrices of the form:
	\[
	\begin{bmatrix}
		1 & \overline{m}_1 & m \\
		0  &  I_{n}  &  \overline{m}_2 \\
		0  & 0 & 1
	\end{bmatrix}
	\]
	with $m \in \ZZ$ and $\overline{m}_{1},\overline{m}_{2} \in \ZZ^{n}$, such that the affine structure induced on $\mathcal{H}_{2n+1}(\RR)/\Gamma$ is integral.
	
	Note that the linear part of the affine representation is given by:
	\begin{align}
		\rho \colon \mathcal{H}_{2n+1}(\RR)  & \longrightarrow \mathrm{GL} (\RR^{2n+1}) \\
		\begin{bmatrix}
			1 & \overline{a} & c \\
			0  &  I_{n}  &  \overline{b} \\
			0  & 0 & 1
		\end{bmatrix} & \longmapsto \begin{bmatrix}
			I_n & 0 & 0 \\
			0 & 1 & \overline{a} \\
			0 & 0 & I_n
		\end{bmatrix}.
	\end{align}
	Let $\{X_1, \ldots, X_n , T , Y_1, \ldots , Y_n\}$ be a basis for the Lie algebra of $\mathcal{H}_{2n+1}(\RR)$. We then have:
	\[
	\rho_{*}(X_i) = \begin{bmatrix}
		0_{n}  & 0 & 0 \\
		0 & 0 & E_i \\
		0 & 0 & 0_{n}
	\end{bmatrix} 
	\]
	where $E_i$ is a row vector of zeros with a $1$ in the $i$-th entry, and $\rho_{*}(Y_i) = \rho_{*}(T) = 0$. Note that:
	\[
	\sum_{j=1}^{n} (\rho_{*}(X_j))_{i,j} + (\rho_{*}(T))_{i,n+1} + \sum_{j=n+2}^{2n+1} (\rho_{*}(Y_{j-(n+1)}))_{i,j} = \sum_{j=1}^{n} (\rho_{*}(X_j))_{i,j} = 0
	\]
	for all $1 \le i \le 2n+1$. 
	
	By Theorem~\ref{mirrorsolvmanifolds}, the dual torus bundles $X$ and $\breve{X}$ form a non-Kähler SYZ mirror pair $(X, \mathsf{\Omega}, \omega)$ and $(\breve{X}, \breve{\omega}, \breve{\mathsf{\Omega}})$.
\end{example}

\vspace{1em}

At this point, a natural question is whether one can produce other families of examples in which Theorem \ref{mirrorsolvmanifolds} could be applied. At first glance, it is very tempting to proceed as follows.
	
\vspace{1em}

Let $G$ be a simply connected Lie group equipped with a complete left-invariant metric $\langle \cdot, \cdot \rangle$. Assume that the Levi-Civita connection $\nabla$ of this metric is flat. Then, it defines a complete left-invariant affine structure on $G$ (completeness follows from the Hopf-Rinow Theorem). Let $\{X_1,\ldots,X_n\}$ be an orthonormal basis for $\mathfrak{g}$. Because the connection is compatible with the metric, we have the following:
\[
\langle \nabla_{X_i} X_i , X_j \rangle = - \langle X_i, \nabla_{X_i} X_j\rangle 
\]
for all $1\le i,j \le n$. In terms of Christoffel symbols, this implies:
\[
\Gamma^{j}_{i,i} = - \Gamma^{i}_{i,j} = \Gamma^{i}_{j,i}.
\]
Setting $j=i$, this immediately forces $\Gamma^{i}_{i,i} = 0$ for all $1 \le i \le n$. Using the fact that the connection gives a complete left-invariant affine structure on $G$, we see that $G$ is unimodular. Hence:
\[
\sum_{j=1}^{n} \Gamma^{i}_{j,j} = \sum_{j=1}^{n} \Gamma^{j}_{i,j} = 0
\]
for all $ 1\le i \le n$. Thus, the algebraic hypotheses of Theorem~\ref{mirrorsolvmanifolds} are perfectly satisfied.

If $G$ admits a lattice $\Gamma$ such that the induced affine structure on $G/ \Gamma$ is integral, Theorem~\ref{mirrorsolvmanifolds} would immediately yield a non-Kähler SYZ mirror pair:
\[
\xymatrix{
	(X, \mathsf{\Omega}, \omega) \ar[dr]_{\pi} & & (\breve{X}, \breve{\omega}, \breve{\mathsf{\Omega}}) \ar[dl]^{\breve{\pi}} \\
	& G/\Gamma &
}
\]  
However, a classical result of Milnor~\cite[Corollary 3.2, Theorem 3.3]{Milnor1976} dictates that under these exact conditions, $G$ must have a commutative Lie algebra. Thus $G \cong \RR^{n}$, which yields no new geometrically interesting examples. 

Nevertheless, if we relax the Riemannian signature, we can prove the following:

\begin{theo}\label{thm:lorentzian}
	Let $G$ be a simply connected Lie group endowed with a left-invariant Lorentzian metric $\langle \cdot, \cdot \rangle$. Assume that the Levi-Civita connection $\nabla$ is flat and geodesically complete. Furthermore, assume that there exists a lattice $\Gamma$ such that the induced affine structure on $G/ \Gamma$ is integral. Then, the construction of the dual torus bundle presented in Section~\ref{affine},
	\[
	\xymatrix{
		(X, \mathsf{\Omega}, \omega) \ar[dr]_{\pi} & & (\breve{X}, \breve{\omega}, \breve{\mathsf{\Omega}}) \ar[dl]^{\breve{\pi}} \\
		& G/\Gamma &
	}
	\]
	forms a non-Kähler SYZ mirror pair.
\end{theo}
\begin{proof}
	Following the same algebraic manipulations as above, we have:
	\[
	\sum_{j=1}^{n} \Gamma^{i}_{j,j} = \sum_{j=1}^{n} \Gamma^{j}_{i,j} = 0
	\]
	and $\Gamma^{i}_{i,i} = 0$ for all $1\le i \le n$. Thus, the conditions of Theorem~\ref{mirrorsolvmanifolds} are satisfied, completing the proof.   
\end{proof}

\vspace{1em}

From a different perspective, we can provide a complete characterization of all non-Kähler SYZ mirror pairs arising from nilpotent Lie groups.

\begin{theo}\label{mirror_nilmanifolds}
	Let $G$ be a simply connected nilpotent Lie group with a complete left-invariant affine structure. Denote by $\rho$ the linear part of the affine representation and let $\{X_1,\ldots , X_n\}$ be a basis for $\mathfrak{g}$. Then, $G$ produces a non-Kähler SYZ mirror pair if and only if:
	\begin{itemize}
		\item[(1)] $C^{i}_{j,k} \in \QQ$ for all $1 \le i,j,k \le n$;
		\item[(2)] $(\rho_{*}(X_i))_{i,i} = \frac{1}{2}\sum_{j=1}^{n} (\rho_{*}(X_j))_{i,j}$ for all $1\le i \le n$;
		\item[(3)] There exists a lattice $\gamma \subset \mathfrak{g}$ such that $\exp(\rho_{*}(X))$ is conjugate to an integer matrix for all $X \in \gamma$.
	\end{itemize}
	Furthermore, the lattice in $G$ is explicitly given by $\Gamma := \exp(\gamma)$, making the induced affine structure on $G/ \Gamma$ integral and yielding the mirror pair:
	\[
	\xymatrix{
		(X, \mathsf{\Omega}, \omega) \ar[dr]_{\pi} & & (\breve{X}, \breve{\omega}, \breve{\mathsf{\Omega}}) \ar[dl]^{\breve{\pi}} \\
		& G/\Gamma &
	}
	\]
\end{theo}
\begin{proof}
	By Mal'cev's criterion~\cite[Theorem 2.12]{Raghunathan1972}, knowing that the structure constants of $\mathfrak{g}$ are rational ensures the existence of a lattice $\Gamma$ in $G$. Furthermore, every lattice in a simply connected nilpotent Lie group $G$ is obtained by exponentiating some lattice in its Lie algebra $\mathfrak{g}$.
	
	We only need to verify whether the affine structure induced in $G/\Gamma$ is integral. From the third condition, for $\exp(X) \in \Gamma$, the element
	\[
	\rho(g) = \rho(\exp(X)) = \exp(\rho_{*}(X))
	\]
	is conjugate to an integer matrix. In other words, the restriction $\restr{\rho}{\Gamma}$ is conjugate to a subgroup of $\mathrm{GL}(n,\ZZ)$.
	
	Finally, invoking the second condition alongside Theorem~\ref{mirrorsolvmanifolds} completes the proof.
\end{proof}

\vspace{1em}

There is a complete characterization of complete left-invariant affine structures on simply connected nilpotent Lie groups~\cite{kim1986complete} across several dimensions. A plethora of examples satisfying Theorem \ref{mirrorsolvmanifolds} can be systematically built, akin to the one below: 

\begin{example}\label{ex:plethora}
	Let $\mathfrak{g}$ be a nilpotent Lie algebra of dimension $4$. Denoting by $\{X_1,\ldots, X_4\}$ a basis of $\mathfrak{g}$, consider the following LSA\footnote{There is an equivalence between Left Symmetric Algebra (LSA) structures and left-invariant affine structures. Thus, prescribing an affine structure is equivalent to defining an LSA product, see \cite[p. 376]{kim1986complete}.} structure:
	\begin{align}
		X_2 \cdot X_3 &= X_1, & X_3 \cdot X_3 &= X_2, & X_1 \cdot X_4 &= X_2, \\
		X_4 \cdot X_1 &= X_2, & X_2 \cdot X_4 &= -X_3. & &
	\end{align}
	Let $\rho$ be the affine representation associated with the complete affine structure (induced by the LSA product). We have:
	\[
	\rho_{*}(X)Y = X \cdot Y.
	\]
	This yields the following matrix representations:
	\begin{align}
		\rho_{*}(X_1) &= \begin{bmatrix}
			0  &  0 &  0  &  0 \\
			0  &  0 &  0  &  0 \\
			0  &  1 &  0  &  0 \\
			0  &  0  &  0  &  0  
		\end{bmatrix}, & 
		\rho_{*}(X_2) &= \begin{bmatrix}
			0  &  0  &  0   &  0 \\
			0  &  0  &  0  &  0  \\
			1  &  0  &  0  &  0 \\
			0  &  0  &  -1  &  0
		\end{bmatrix}, \\
		\rho_{*}(X_3) &= \begin{bmatrix}
			0  &  0  &  0  &  0 \\
			0  &  0  &  0  &  0 \\
			0  &  1  &  0  &  0 \\
			0  &  0  &  0  &  0
		\end{bmatrix}, & 
		\rho_{*}(X_4) &= \begin{bmatrix}
			0  &  1  &  0  &  0 \\
			0  &  0  &  0  &  0  \\
			0  &  0  &  0  &  0  \\
			0  &  0  &  0  &  0  
		\end{bmatrix}.
	\end{align}
	Note that conditions (1) and (2) of Theorem \ref{mirror_nilmanifolds} are identically satisfied. Let $\gamma = \mathrm{span}_{\ZZ}\{X_1 ,X_2 , X_3,X_4\}$. Then, for $X = \sum_{i=1}^{4} n_i X_i \in \gamma$ with integers $n_i \in \ZZ$, we have:
	\[
	\rho_{*}(X) = \begin{bmatrix}
		0  &  n_4  &  0  &  0 \\
		0  &  0    &  0  &  0 \\
		n_2 &  n_3 &  0  &  0 \\
		0   &  n_1 &  -n_2 & 0 
	\end{bmatrix}.
	\]
	It can be verified that $\exp(\rho_{*}(X))$ is conjugate to an integer matrix. Therefore, by Theorem~\ref{mirror_nilmanifolds}, there is a non-Kähler SYZ mirror pair with base $G/\Gamma$, where $\Gamma = \exp(\gamma)$.
\end{example}
	
	\ 
	
\section{On the Tseng-Yau cohomology}

Let $(\breve{X}, \breve{\omega})$ be a symplectic manifold and let $\boldsymbol{\alpha}$ denote the bivector field dual to the symplectic form $\breve{\omega}$. In \cite{Tseng_2012}, Tseng and Yau introduced two cohomology theories that are sensitive to the symplectic structure. Specifically, they are invariant under symplectomorphisms but not necessarily under diffeomorphisms, and they are finite-dimensional on compact manifolds. 

Let $(\Omega^{\bullet}(\breve{X}, \CC), \d)$ be the de Rham complex. Set $\Lambda = \iota(\boldsymbol{\alpha})$ as the interior product with $\boldsymbol{\alpha}$, and let $\dL \colon \Omega^{k}(\breve{X}, \CC) \to \Omega^{k-1}(\breve{X}, \CC)$ be given by $\dL = [\d, \Lambda]$. It is straightforward to check that:
\[
\d\dL = -\dL\d, \quad (\d+\dL) \circ \d\dL = 0.
\]

Consider the complex:
\begin{equation}\label{eq:complex1}
	\Omega^{k}(\breve{X}, \CC) \xrightarrow{\d\dL} \Omega^{k}(\breve{X}, \CC) \xrightarrow{\d +\dL} \Omega^{k+1}(\breve{X}, \CC) \oplus \Omega^{k-1}(\breve{X}, \CC).
\end{equation}
The \emph{Tseng-Yau cohomology} is defined as:
\begin{equation}\label{eq:TSYAU}
	H^{k}_{\d+\dL}(\breve{X}) = \frac{\ker\left(\d+\dL \colon \Omega^{k}(\breve{X}, \CC) \longrightarrow \Omega^{k+1}(\breve{X}, \CC) \oplus \Omega^{k-1}(\breve{X}, \CC) \right)}{\im\left( \d\dL \colon \Omega^{k}(\breve{X}, \CC) \longrightarrow \Omega^{k}(\breve{X}, \CC) \right)}.
\end{equation}

Now reverse the arrows in the complex \eqref{eq:complex1}:
\begin{equation}\label{eq:complex2}
	\begin{tikzcd}[row sep=0.1em]
		\Omega^{k+1}(\breve{X}, \CC) \arrow[rd, "\dL"] & & \\
		\oplus & \Omega^{k}(\breve{X}, \CC) \arrow[r, "\d\dL"] & \Omega^{k}(\breve{X}, \CC) \\
		\Omega^{k-1}(\breve{X}, \CC) \arrow[ru, "\d"'] & &
	\end{tikzcd}
\end{equation}
The cohomology of the complex \eqref{eq:complex2} is defined as
\begin{equation}
	H^{k}_{\d\dL}(\breve{X}) = \frac{\ker\left( \d\dL \colon \Omega^{k}(\breve{X}, \CC) \longrightarrow \Omega^{k}(\breve{X}, \CC) \right)}{\left( \im \,\d + \im\,\dL \right)\cap \Omega^{k}(\breve{X}, \CC)}.
\end{equation}

It turns out that whenever $\breve X$ is compact, both cohomologies coincide:

\begin{theo_with_name}{Tseng, Yau~\cite[Corollary 3.25]{Tseng_2012}} 
    Let $(\breve{X},\breve{\omega})$ be a compact symplectic manifold of dimension $2n$. Then the Hodge star operator induces the following isomorphism:
	\[
	H^{k}_{\d+\dL}(\breve{X}) \simeq H^{2n-k}_{\d\dL}(\breve{X})
	\]
	for all $ 0 \le k \le 2n$.
\end{theo_with_name}

\vspace{1em}

Drawing inspiration from Brylinski's definition of Poisson cohomology \cite{Brylinski1988} and building upon the foundational work of Connes \cite{Connes}, Kontsevich \cite{Kontsevich2008}, and Nest and Tsygan \cite{NestTsyganCyclic, Kontsevich2008}, this section introduces a bicomplex construction designed to recover Tseng-Yau cohomology. Our ultimate aim is to explore its potential connection to the periodic cyclic homology of a given dg-algebra; although this correspondence remains a subject for future work. Furthermore, we prove that our constructions are compatible with the Fourier-Mukai transform and, by extension, with non-K\"ahler SYZ mirror symmetry in the sense of Lau-Tseng-Yau.

\vspace{1em}

\subsection{Brylinski's and Tseng-Yau's bicomplexes}
Consider the following operators on the de Rham complex $(\Omega^{\bullet}(\breve{X}, \CC), \d)$:
\begin{itemize}
	\item $B := \d \colon \Omega^{k}(\breve{X}, \CC) \to \Omega^{k+1}(\breve{X}, \CC)$, the standard de Rham differential (degree $+1$).
	\item $b := \dL \colon \Omega^{k}(\breve{X}, \CC) \to \Omega^{k-1}(\breve{X}, \CC)$, where $\dL = [\d, \Lambda]$ (degree $-1$).
\end{itemize}
Note that $\dL$ coincides with the Koszul differential \cite{Brylinski1988}. These operators $(B, b)$ naturally define the $(\Omega^{\bullet}(\breve{X}, \CC), b, B)$ \emph{Brylinski's bicomplex}, since:
\[
B^2 = \d^2 = 0, \quad b^2 = (\dL)^2 = 0, \quad \text{and} \quad Bb + bB = \d\dL + \dL\d = 0.
\]

Let $u$ be a formal variable of degree $+2$. We consider the ``positive'' complex $C_1$:
\[
C^\bullet_1 = \left(\Omega^{\bullet}(\breve{X}, \CC)((u)), D_1 = \d + u\dL\right)
\] 
with a total differential $D_1$ of degree $+1$ satisfying $D_1^2 = 0$. It can be straightforwardly proven\footnote{We have chosen to omit this proof as it would deviate too much from our goals.} that there exists an isomorphism of graded vector spaces:
\[
H(C^\bullet_1) \cong H_{dR}(\breve{X}, \CC)((u)).
\]
 
Moreover, let $A = C^{\infty}(\breve{X}, \CC)$ be the Fréchet algebra of smooth functions on $(\breve{X}, \breve{\omega})$. We can actually check that there is an isomorphism of graded vector spaces:
\begin{equation}\label{eq:ISO} 
H(C^\bullet_1) \cong HP_\bullet(A). 
\end{equation}
Here, $HP_\bullet(A)$ denotes Connes' periodic cyclic homology for topological algebras, i.e., the underlying continuous Hochschild complex is constructed taking into account the Fréchet topology of $A$, utilizing the completed projective tensor product $C^\infty(\breve{X}) \hat{\otimes}_\pi C^\infty(\breve{X}) \cong C^\infty(\breve{X} \times \breve{X})$ \cite{Connes, NestTsyganCyclic}.

We use the isomorphism \eqref{eq:ISO} as inspiration to define the \emph{Tseng-Yau bicomplex} $C^{\bullet} =(\Omega^{\bullet}(\breve{X}, \CC)((u)), \d + u \dL, \d\dL)$:
\begin{center}
	\begin{tikzpicture}
		\tikzset{
			hdiff/.style={->, font=\small}, 
			vdiff/.style={->, font=\small}  
		}
		
		\def\colsep{3.0} 
		\def\rowsep{1.5} 
		
		\node (Cm12) at (0, 2*\rowsep) {$\cdots$};
		\node (C02) at (1*\colsep, 2*\rowsep) {$C^{k}$};
		\node (C12) at (2*\colsep, 2*\rowsep) {$C^{k+1}$};
		\node (C22) at (3*\colsep, 2*\rowsep) {$C^{k+2}$};
		\node (Cp12) at (4*\colsep, 2*\rowsep) {$\cdots$};
		
		\node (Cm11) at (0, 1*\rowsep) {$\cdots$};
		\node (C01) at (1*\colsep, 1*\rowsep) {$C^k$};
		\node (C11) at (2*\colsep, 1*\rowsep) {$C^{k+1}$};
		\node (C21) at (3*\colsep, 1*\rowsep) {$C^{k+2}$};
		\node (Cp11) at (4*\colsep, 1*\rowsep) {$\cdots$};
		
		\node (Cm10) at (0, 0) {$\cdots$};
		\node (C00) at (1*\colsep, 0) {$C^k$};
		\node (C10) at (2*\colsep, 0) {$C^{k+1}$};
		\node (C20) at (3*\colsep, 0) {$C^{k+2}$};
		\node (Cp10) at (4*\colsep, 0) {$\cdots$};
		
		\draw[hdiff] (Cm12) -- (C02) node[midway, above] {$\d+u\dL$};
		\draw[hdiff] (C02) -- (C12) node[midway, above] {$\d+u\dL$};
		\draw[hdiff] (C12) -- (C22) node[midway, above] {$\d+u\dL$};
		\draw[hdiff] (C22) -- (Cp12) node[midway, above] {$\d+u\dL$};
		
		\draw[hdiff] (Cm11) -- (C01) node[midway, above] {$\d+u\dL$};
		\draw[hdiff] (C01) -- (C11) node[midway, above] {$\d+u\dL$};
		\draw[hdiff] (C11) -- (C21) node[midway, above] {$\d+u\dL$};
		\draw[hdiff] (C21) -- (Cp11) node[midway, above] {$\d+u\dL$};
		
		\draw[hdiff] (Cm10) -- (C00) node[midway, above] {$\d+u\dL$};
		\draw[hdiff] (C00) -- (C10) node[midway, above] {$\d+u\dL$};
		\draw[hdiff] (C10) -- (C20) node[midway, above] {$\d+u\dL$};
		\draw[hdiff] (C20) -- (Cp10) node[midway, above] {$\d+u\dL$};
		
		\draw[vdiff] (C02) -- (C01) node[midway, right] {$\d\dL$}; 
		\draw[vdiff] (C01) -- (C00) node[midway, right] {$\d\dL$}; 
		
		\draw[vdiff] (C12) -- (C11) node[midway, right] {$\d\dL$}; 
		\draw[vdiff] (C11) -- (C10) node[midway, right] {$\d\dL$}; 
		
		\draw[vdiff] (C22) -- (C21) node[midway, right] {$\d\dL$};
		\draw[vdiff] (C21) -- (C20) node[midway, right] {$\d\dL$};
		
	\end{tikzpicture}
\end{center}

\vspace{1em}

\subsection{New cohomologies}

The Tseng-Yau bicomplex have to it associated four distinct cohomologies: 
\begin{enumerate}
	\item $\d\dL$ Cohomology:
	\[
    H_{s}^{k}(C^{\bullet}) := \frac{\ker(\d\dL \colon C^{k}\longrightarrow C^{k})}{\im(\d\dL \colon C^{k}\longrightarrow C^{k})}
    \]
	
	\item Periodic Cyclic Cohomology (P.C.):
	\[
    HP^{k}(C^{\bullet}) := \frac{\ker(\d+u\dL \colon C^{k}\longrightarrow C^{k+1})}{\im(\d+u\dL \colon C^{k-1}\longrightarrow C^{k})}
    \]
	
	\item Bott-Chern Cohomology:
	\[
    H_{\d + \dL}^{k}(C^{\bullet}) := \frac{\ker(\d+u\dL \colon C^{k}\longrightarrow C^{k+1})}{\im(\d\dL \colon C^{k}\longrightarrow C^{k})}
    \]
	
	\item Aeppli Cohomology:
	\[
    H_{\d\dL}^{k}(C^{\bullet}) := \frac{\ker(\d \dL \colon C^{k}\longrightarrow C^{k})}{\im(\d+ u\dL \colon C^{k-1}\longrightarrow C^{k})}
    \]
\end{enumerate}

\vspace{1em}

The cohomologies $H_{\d + \dL}^{k}(C^{\bullet})$ and $H_{\d\dL}^{k}(C^{\bullet})$ parallel the classical Tseng-Yau cohomologies $H^{k}_{\d + \dL}(\breve{X})$ and $H^{k}_{\d \dL}(\breve{X})$ developed in \cite{Tseng_2012}. The primary distinction lies in the fact that $H_{\d +\dL}(C^{\bullet})$ captures mixed-degree forms. Specifically, classes in $H_{\d +\dL}^{k}(C^{\bullet})$ take the form $[A_{2i} - A_{2(i-1)}]$ where $A_{i} \in \Omega^{k-2i}(\breve{X}, \CC)$, whereas classes in $H^{k}_{\d + \dL}(\breve{X})$ consist of forms of pure degree $k$. More precisely, we have the following isomorphisms:

\begin{theo}\label{BC_cohomology}
	Let $(\breve{X},\breve{\omega})$ be a symplectic manifold. There exist subspaces $Q_{\mathrm{even}} \subset H_{\d + \dL}^{\mathrm{even}}(C^{\bullet})$ and $Q_{\mathrm{odd}} \subset H_{\d + \dL}^{\mathrm{odd}}(C^{\bullet})$ such that we have the following isomorphisms of graded vector spaces:
	\begin{align}
		H_{\d+\dL}^{\mathrm{even}}(C^{\bullet}) &\cong \left(\bigoplus_{m~\mathrm{even}}H_{\d + \dL}^{m}(\breve{X}, \CC)\right) \oplus Q_{\mathrm{even}}, \\
		H_{\d+\dL}^{\mathrm{odd}}(C^{\bullet}) &\cong \left(\bigoplus_{m~\mathrm{odd}}H_{\d + \dL}^{m}(\breve{X} , \CC)\right) \oplus Q_{\mathrm{odd}}.
	\end{align}
\end{theo}

\begin{proof} 
	We detail the proof for the even case; the odd case follows analogously.

	Using the canonical identification between $C^{k}$ and $\bigoplus_{i \in \ZZ} \Omega^{k-2i}(\breve{X}, \CC)$, we observe that:
	\[
	\ker(\d+u\dL \colon C^{k}\to C^{k+1}) \cong \ker\left( \d+\dL \colon \bigoplus_{i\in \mathbb{Z}}\Omega^{k-2i}(\breve{X} ,\CC) \longrightarrow \bigoplus_{i\in \mathbb{Z}}\Omega^{k+1-2i}(\breve{X}, \CC) \right)
	\]
	and
	\[
	\im(\d\dL \colon C^{k}\to C^{k}) \cong \bigoplus_{m~\mathrm{even}}\im(\d\dL \colon \Omega^{m}(\breve{X}, \CC) \to\Omega^{m}(\breve{X}, \CC)).
	\]
	Consequently, we define the quotient space $V$ as:
	\begin{equation}\label{identification}
		V := \frac{\ker\left( \d+\dL \colon \bigoplus_{i\in \mathbb{Z}}\Omega^{k-2i}(\breve{X}, \CC) \longrightarrow \bigoplus_{i\in \mathbb{Z}}\Omega^{k+1-2i}(\breve{X}, \CC) \right)}{\bigoplus_{m~\mathrm{even}}\im(\d\dL \colon \Omega^{m}(\breve{X}, \CC)\to\Omega^{m}(\breve{X}, \CC))}.
	\end{equation}
	
	Let $A = \sum_{k=0}^{n}[A_{2k}] \in \bigoplus_{m~\mathrm{even}}H_{\d + \dL}^{m}(\breve{X}, \CC)$, with $[A_{2k}] \in H_{\d + \dL}^{2k}(\breve{X}, \CC)$ for all $0 \le k \le n$. Our goal is to prove that $\bigoplus_{m~\mathrm{even}}H^{m}_{\d + \dL}(\breve{X}, \CC) \subset V$.
	
	Note that $A_{2k} \in \ker(\d) \cap \ker(\dL) \cap \Omega^{2k}(\breve{X}, \CC)$. Thus, applying the operator yields:
	\begin{align}
		(\d + \dL)\left( \sum_{k=0}^{n}A_{2k} \right) &= \sum^{n}_{k=1} \left( \d A_{2k} + \dL A_{2(k-1)} \right) = 0 \\
		&\implies \sum_{k=0}^{n}A_{2k} \in \ker\left( \d+\dL \colon \bigoplus_{i\in \mathbb{Z}}\Omega^{k-2i}(\breve{X}, \CC) \longrightarrow \bigoplus_{i\in \mathbb{Z}}\Omega^{k+1-2i}(\breve{X}, \CC) \right).
	\end{align}
	This implies that $\sum_{k=0}^{n}A_{2k}$ belongs to the kernel of $\d+\dL$ on the direct sum.
	
	Next, suppose $A,A' \in \bigoplus_{m~\mathrm{even}}H^{m}_{\d +\dL}(\breve{X}, \CC)$ such that $A - A' = \sum^{n}_{k=0} \d \dL B_{2k}$, where $B_{2k} \in \Omega^{2k}(\breve{X}, \CC)$ for all $k$. It immediately follows that:
	\[
	A-A' \in \bigoplus_{m~\mathrm{even}}\im(\d\dL \colon \Omega^{m}(\breve{X}, \CC)\to\Omega^{m}(\breve{X}, \CC))
	\]
	so the inclusion mapping $i \colon \bigoplus_{m~\mathrm{even}}H^{m}_{\d +\dL}(\breve{X}, \CC) \hookrightarrow V$ is well-defined. We now show that $i$ is injective.
	
	Assume $i(A) = \left[ \sum^{n}_{k=0}A_{2k} \right] = 0$. This forces:
	\begin{align}
		\sum^{n}_{k=0}A_{2k} &\in \bigoplus_{m~\mathrm{even}} \im(\d\dL \colon \Omega^{m}(\breve{X}, \CC) \to \Omega^{m}(\breve{X}, \CC)) \\
		&\implies A_{2k} \in \im(\d\dL \colon \Omega^{2k}(\breve{X}, \CC) \to \Omega^{2k}(\breve{X}, \CC))
	\end{align}
	for each $k$. Thus, $A=0$, which implies that $\bigoplus_{m~\mathrm{even}} H^{m}_{\d + \dL}(\breve{X}, \CC) \subset V$.
	
    By completing the basis of $V$, we can find a complementary subspace $Q_{\mathrm{even}} \subset V$ such that:
	\[
	V = \left( \bigoplus_{m~\mathrm{even}}H^{m}_{\d + \dL}(\breve{X}, \CC) \right) \oplus  Q_{\mathrm{even}}.
	\]
	Equation \eqref{identification} then yields the desired isomorphism:
	\[
	H^{\mathrm{even}}_{\d+\dL}(C^{\bullet}) \cong \left( \bigoplus_{m~\mathrm{even}}H^{m}_{\d + \dL}(\breve{X}, \CC) \right) \oplus  Q_{\mathrm{even}}.
	\]
\end{proof}

A drawback of the cohomology $H_{\d+\dL}(C^\bullet)$ is that the $Q$-term is large, making the space infinite-dimensional even on compact manifolds, as the following example shows.

\begin{example}
	Let $(\breve{X}, \breve{\omega})$ be a Kähler manifold of real dimension $2$. Then
	\begin{align}
	H^{\mathrm{odd}}_{\d+\dL}(C^{\bullet}) &= H^{1}_{\d+\dL}(\breve{X}, \CC)\\
	H^{\mathrm{even}}_{\d+\dL}(C^{\bullet}) &\simeq \left\{ f + A_{2} \;\middle|\; \d(f + \Lambda A_{2}) =0, \, f \in C^{\infty}(\breve{X}, \CC), \, A_{2} \in \Omega^{2}(\breve{X}, \CC) \right\}
	\end{align}
	where $\Lambda = \iota(\breve{\omega}^{-1})$.
\end{example}

\vspace{1em}

However, it is notable that if we restrict this cohomology to \emph{primitive forms}, we recover the primitive version of the Tseng-Yau cohomology; this is the content of Theorem \ref{prim_BC}. Let us proceed by recalling some needed concepts.

\vspace{1em}

Recall the definition of the \emph{dual Lefschetz operator}:
\begin{align}
	\Lambda \colon \Omega^{k}(\breve{X}, \CC) & \longrightarrow \Omega^{k-2}(\breve{X}, \CC) \\
	A & \longmapsto \Lambda A := \iota(\boldsymbol{\alpha})A.
\end{align}

\begin{defin}\label{def:primitive-forms}
	A \emph{primitive form} $B \in \Omega^{m}(\breve{X}, \CC)$ is an $m$-form such that $\Lambda(B) = 0$. In what follows, we denote by $\mathcal{P}^{m}(\breve{X}, \CC)$ the space of primitive $m$-forms, and by $\mathcal{P}(\breve{X}, \CC) = \bigoplus_{i=0}^{2n} \mathcal{P}^{i}(\breve{X}, \CC)$ the space of all primitive forms.
\end{defin}

\vspace{1em}

Tseng and Yau defined in \cite{Tseng_2012} a primitive version of $H_{\d +\dL}(\breve{X})$:
\[
PH_{\d +\dL}^{k}(\breve{X}) := \frac{\ker\left( \d + \dL \colon \Omega^{k}(\breve{X}, \CC) \longrightarrow \Omega^{k+1}(\breve{X}, \CC) \oplus \Omega^{k-1}(\breve{X}, \CC) \right) \cap \mathcal{P}^{k}(\breve{X}, \CC)}{\im\left( \d\dL \colon \Omega^{k}(\breve{X}, \CC) \longrightarrow \Omega^{k}(\breve{X}, \CC) \right) \cap \mathcal{P}^{k}(\breve{X}, \CC)}.
\]
Similarly, we can define the notion of primitive cohomology for $HP(C^{\bullet})$:

\begin{defin}\label{def:primitive-complex}
	For the complex $C^{\bullet}=(\Omega(\breve{X}, \CC)((u)), \d + u\dL, \d\dL)$, using the identification $C^{k} \simeq \bigoplus_{i \in \ZZ}\Omega^{k-2i}(\breve{X}, \CC)$, one can define the \emph{primitive Bott-Chern cohomology} as:
	\[
	PH_{\d+\dL}^{k}(C^{\bullet}) := \frac{\ker\left( \d+\dL \colon \bigoplus_{m \in \ZZ}\Omega^{k-2m}(\breve{X}, \CC)\longrightarrow \bigoplus_{m \in \ZZ}\Omega^{k-2m}(\breve{X}, \CC) \right)\cap \mathcal{P}(\breve{X}, \CC) }{\left(  \bigoplus_{m \in \ZZ}\im\left(\d\dL \colon \Omega^{k-2m}(\breve{X}, \CC)\longrightarrow\Omega^{k-2m}(\breve{X}, \CC) \right) \right) \cap \mathcal{P}(\breve{X}, \CC) }.
	\]
\end{defin}

\begin{prop}\label{cor:better-recognized}
For the primitive Bott-Chern cohomology (Definition \ref{def:primitive-complex}), we have for each $k\geq 0$:
	\[
	PH_{\d+\dL}^{k}(C^{\bullet}) = \frac{\ker\left( \d+\dL \colon \bigoplus_{m \in \ZZ}\Omega^{k-2m}(\breve{X}, \CC)\longrightarrow \bigoplus_{m \in \ZZ}\Omega^{k-2m}(\breve{X}, \CC) \right)\cap \mathcal{P}(\breve{X}, \CC) }{\bigoplus_{m \in \ZZ}\d\dL \left(\mathcal{P}^{k-2m}(\breve{X}, \CC) \right)}.
	\] 
\end{prop}
\begin{proof}
    This is a direct consequence of \cite[Lemma 3.9]{Tseng_2012}.
\end{proof}

\begin{theo}\label{prim_BC}
	Let $(\breve{X}, \breve{\omega})$ be a symplectic manifold of dimension $2n$. Then:
	\[
	PH^{k}_{\d+\dL}(C^{\bullet}) \cong \begin{cases}
		\bigoplus_{m ~\mathrm{even}} PH_{\d+\dL}^{m}(\breve{X}, \CC), & \text{if } k \text{ is even,} \\[2ex]
		\bigoplus_{m ~\mathrm{odd}} PH_{\d+\dL}^{m}(\breve{X}, \CC), & \text{if } k \text{ is odd.}
	\end{cases}
	\]
\end{theo}

\begin{proof}
	We present the proof for the even case, as the odd case is analogous.
	
	Let $\alpha = \sum^{n}_{i=0} B_{2i} \in \bigoplus_{m~\mathrm{even}} \mathcal{P}^{m}(\breve{X}, \CC)$, where $B_{2i} \in \mathcal{P}^{2i}(\breve{X}, \CC)$ for $0 \le i \le n$. Observe that applying the operator $(\d + \dL)$ yields:
	\[
	(\d + \dL)(\alpha) = \sum^{n-1}_{i=1} \left( \d B_{2i} + \dL B_{2(i+1)} \right).
	\]
	For primitive forms, we have the identity $\dL B_{2k} = \Lambda \d B_{2k}$. Substituting this into the sum gives:
	\[
	(\d+\dL)(\alpha) = \sum^{n-1}_{i=1} \left( \d B_{2i} + \Lambda \d B_{2(i+1)} \right).
	\]
	Therefore, $(\d + \dL)(\alpha) = 0$ if and only if:
	\[
	\d B_{2i} + \Lambda \d B_{2(i+1)} = 0 \quad \text{for all } 1 \le i \le n-1.
	\]
	We can show that this system of equations holds if and only if $\d B_{2i} = 0$ for all $0 \le i \le n$. We proceed by descending induction on the index $i$. Because $\breve{X}$ has dimension $2n$, any form of degree strictly greater than $2n$ must vanish; thus, $B_{2(n+1)} = 0$. Consequently, for $i = n$, the equation reduces to $\d B_{2n} = 0$. Assuming $\d B_{2(i+1)} = 0$ for some $i$, the relation $\d B_{2i} + \Lambda \d B_{2(i+1)} = 0$ implies $\d B_{2i} = 0$. By induction, $\d B_{2i} = 0$ for all $0 \le i \le n$.
	
	It follows that $(\d + \dL)(\alpha) = 0$ if and only if $\d B_{2i} = 0$ for all $0 \le i \le n$. This implies that the kernel of $(\d + \dL)$ on even-degree primitive forms decomposes completely:
	\begin{align}
		\ker&\left( \d+\dL \colon \bigoplus_{m~\mathrm{even}}\Omega^{m}(\breve{X}, \CC) \longrightarrow \bigoplus_{m~\mathrm{odd}}\Omega^{m}(\breve{X}, \CC) \right) \cap \mathcal{P}(\breve{X}, \CC) \\
		&= \bigoplus_{m~\mathrm{even}} \ker\left( \d \colon \Omega^{m}(\breve{X}, \CC) \to \Omega^{m+1}(\breve{X}, \CC) \right) \cap \mathcal{P}(\breve{X},\CC).
	\end{align}
	Taking the quotient by the image of $\d\dL$, we obtain the desired isomorphism:
	\[
	PH^{\mathrm{even}}_{\d+\dL}(C^{\bullet}) \cong \frac{\displaystyle \bigoplus_{m~\mathrm{even}} \ker\left( \d \colon \Omega^{m}(\breve{X}, \CC) \to \Omega^{m+1}(\breve{X}, \CC) \right) \cap \mathcal{P}(\breve{X}, \CC)}{\displaystyle \bigoplus_{m~\mathrm{even}} \d\dL\left( \mathcal{P}^{m}(\breve{X}, \CC) \right)} = \bigoplus_{m~\mathrm{even}} PH_{\d+\dL}^{m}(\breve{X}, \CC).
	\]
\end{proof}

    \ 

\subsection{Complex cohomologies}

Here we build on the discussion established in the previous section analogously aiming at obtaining a ``mirror cohomology'' to $H_{\d+\dL}(C^{\bullet})$ for a given non-K\"ahler SYZ mirror pair.

Let $X$ be a complex manifold of dimension $2n$, and consider the complex:
\[
G^{\bullet} \colon \quad \dots \xrightarrow{\partial} \Omega^{0, n}(X) \xrightarrow{\partial} \Omega^{0,n-1}(X) \oplus\Omega^{1,n}(X) \xrightarrow{\partial} \Omega^{0,n-2}(X) \oplus \Omega^{1,n-1}(X)\oplus \Omega^{2,n}(X) \xrightarrow{\partial} \dots
\]
Defining $G^{k} = \bigoplus_{p = 0}^{k} \Omega^{p, n-k+p}(X)$, we equip $G^{\bullet}$ with the bicomplex structure $G^{\bullet} = (G,\partial, \overline{\partial})$. Observe that $\partial(G^{k}) \subset G^{k+1}$ and $\overline{\partial}(G^{k}) \subset G^{k-1}$. Hence, the operator $\partial$ carries a degree of $+1$, while $\overline{\partial}$ carries a degree of $-1$.

Next, let $u$ be a formal variable of degree $+2$. We define the formal complex $\hat{C}^{\bullet} := (G^\bullet((u)) , \partial + u \overline{\partial})$. The total differential $\partial + u \overline{\partial}$ has degree $+1$, and the $k$-th term of the complex is given by:
\[
\hat{C}^{k} = \left\{ \sum_{i \in \mathbb{Z}} u^{i} g_{i} \;\middle|\; g_i \in G^{k-2i} \text{ for all } i \in \mathbb{Z} \right\}.
\]
Furthermore, we can naturally identify $\hat{C}^{k}$ with the direct sum:
\[
\hat{C}^{k} \cong \bigoplus_{i \in \mathbb{Z}} \bigoplus_{p = 0}^{k-2i} \Omega^{p, n-k+2i+p}(X).
\]
With this structure in place, we define the following four cohomologies:

\begin{enumerate}
	\item $\partial \overline{\partial}$ Cohomology:
	\[
	H_{s}^{k}(\hat{C}^{\bullet}):=\frac{\mathrm{Ker}(\partial \overline{\partial}: \hat{C}^{k}\longrightarrow \hat{C}^{k})}{\mathrm{Im}(\partial \overline{\partial}: \hat{C}^{k}\longrightarrow \hat{C}^{k})}
	\]
	
	\item Periodic Cyclic Cohomology (P.C.):
	\[
	H P^{k}(\hat{C}^{\bullet}):=\frac{\mathrm{Ker}(\partial +u\overline{\partial}:\hat{C}^{k}\longrightarrow \hat{C}^{k+1})}{\mathrm{Im}(\partial +u\overline{\partial}: \hat{C}^{k-1}\longrightarrow \hat{C}^{k})}
	\]
	
	\item Bott-Chern Cohomology:
	\[
	H_{\partial +\overline{\partial}}^{k}(\hat{C}^{\bullet}):=\frac{\mathrm{Ker}(\partial +u\overline{\partial}:\hat{C}^{k}\longrightarrow \hat{C}^{k+1})}{\mathrm{Im}(\partial \overline{\partial}:\hat{C}^{k}\longrightarrow \hat{C}^{k})}
	\]
	
	\item Aeppli Cohomology:
	\[
	H_{\partial + \overline{\partial}}^{k}(\hat{C}^{\bullet}):=\frac{\mathrm{Ker}(\partial \overline{\partial}: \hat{C}^{k}\longrightarrow \hat{C}^{k})}{\mathrm{Im}(\partial+ u\overline{\partial}: \hat{C}^{k-1}\longrightarrow \hat{C}^{k})}
	\]
\end{enumerate}

\vspace{1em}

The proof of Theorem \ref{complex_BC} is analogous to that of Theorem \ref{BC_cohomology} and is therefore omitted.

\begin{theo}\label{complex_BC}
	Let $X$ be a complex manifold of dimension $2n$. There exist subspaces $\hat{Q}_{\mathrm{even}}$ and $\hat{Q}_{\mathrm{odd}}$ in $H_{\partial +\overline{\partial}}^{k}(\hat{C}^{\bullet})$ such that, when $n \equiv 0 \pmod 2$,
	\[
	H_{\partial +\overline{\partial}}^{k}(\hat{C}^{\bullet}) \cong \begin{cases}
		\left( \bigoplus_{m \text{ even}} \bigoplus_{p+q=m} H_{\partial + \overline{\partial}}^{p,q}(X) \right) \oplus \hat{Q}_{\mathrm{even}}, & k \text{ even} \\[2ex]
		\left( \bigoplus_{m \text{ odd}} \bigoplus_{p+q=m} H_{\partial + \overline{\partial}}^{p,q}(X) \right) \oplus \hat{Q}_{\mathrm{odd}}, & k \text{ odd}
	\end{cases}
	\]
	and when $n \equiv 1 \pmod 2$,
	\[
	H_{\partial + \overline{\partial}}^{k}(\hat{C}^{\bullet}) \cong \begin{cases}
		\left( \bigoplus_{m \text{ odd}} \bigoplus_{p+q=m} H_{\partial +\overline{\partial}}^{p,q}(X) \right) \oplus \hat{Q}_{\mathrm{even}}, &  k \text{ even} \\[2ex]
		\left( \bigoplus_{m \text{ even}} \bigoplus_{p+q=m} H_{\partial + \overline{\partial}}^{p,q}(X) \right) \oplus \hat{Q}_{\mathrm{odd}}, & k \text{ odd}
	\end{cases}
	\]
\end{theo}

Not surprisingly, akin to $H_{\d+\dL}(C^{\bullet})$, the cohomology $H_{\partial +\overline{\partial}}(\hat{C}^{\bullet})$ is generically infinite-dimensional:

\begin{example}
	Let $X$ be a complex manifold of complex dimension $1$. Then:
	\begin{align}
	H_{\partial +\overline{\partial}}^{\mathrm{odd}}(\hat{C}^{\bullet}) &\simeq H^{1,1}_{\partial + \overline{\partial}}(X) \oplus H^{0,0}_{\partial + \overline{\partial}}(X)\\
	H^{\mathrm{even}}_{\partial + \overline{\partial}}(\hat{C}^{\bullet}) &\simeq \mathrm{Ker}(\d \colon \Omega^{1}(X) \to \Omega^{2}(X)).
	\end{align}
\end{example}

\vspace{1em}

\subsection{Cohomologies and Non-K\"ahler SYZ}

Given a non-K\"ahler SYZ mirror pair $(X,\mathsf{\Omega}, \omega)$ and $(\breve{X},\breve{\omega},\breve{\mathsf{\Omega}})$, the cohomology $H_{\d + \dL}(\breve{X})$ measures well-behaved supersymmetric deformations of type IIA $\mathrm{SU}(n)$-structures \cite{Gen_cohom}. The Fourier-Mukai transform of differential forms yields an isomorphism:
\[
FT \colon \bigoplus^{k}_{p=0} \Omega_{B}^{p,n-k+p}(X) \overset{\simeq}{\longrightarrow} \Omega^{k}_{B}(\breve{X}, \CC).
\]
Lau, Tseng, and Yau established the following equivalence:

\begin{theo_with_name}{Lau, Tseng, Yau~\cite[Theorem 4.5]{Lau_2015}}\label{Fourier-Mukai_iso}
	The Fourier-Mukai transform gives an isomorphism of bicomplexes:
	\[
	\left( \Omega_{B}(\breve{X}, \CC), \frac{(-1)^{n}\mathbf{i}}{2} \d , \frac{(-1)^{n}\mathbf{i}}{2} \dL  \right) \cong \left( \Omega_{B}(X), \partial , \overline{\partial}  \right).
	\]
\end{theo_with_name}

Moreover, they prove:

\begin{theo_with_name}{Lau, Tseng, Yau~\cite[Theorem 6.7]{Lau_2015}}\label{mirror_Yau_cohomologies}
	Given a non-K\"ahler SYZ mirror pair $(X,\mathsf{\Omega}, \omega)$ and $(\breve{X},\breve{\omega},\breve{\mathsf{\Omega}})$, the Fourier-Mukai transform yields an isomorphism:
	\[
	H^{(n-p,q)^{\Delta}}_{B,\d + \dL}(\breve{X},\CC) \simeq H^{p,q}_{B, \partial + \overline{\partial}} (X),
	\]
	where $\Delta$ is the real polarization on $\breve{X}$ induced by the bundle structure and
    \begin{itemize}
    \item 
\[
H^{p,q}_{B,\partial + \overline{\partial}}(X) := \frac{\mathrm{Ker} (\partial +\overline{\partial}) \cap \Omega^{p,q}_B (X)}{\mathrm{Im}(\partial \overline{\partial}) \cap \Omega^{p,q}_{B}(X)}.
\]
is the basic Bott-Chern cohomology, and
\item 
\[
H_{B,\d +\dL}^{(p,q)^{\Delta}}(\breve{X}) := \frac{\mathrm{Ker}(\d+\dL) \cap \Omega^{(p,q)^{\Delta}}_{B}(\breve{X})}{\mathrm{Im}(\d\dL) \cap \Omega^{(p,q)^{\Delta}}_{B}(\breve{X})}.
\]
is the $\Delta$-refined Tseng-Yau cohomology.
\end{itemize}
\end{theo_with_name}

Although cohomologies $H_{\d+\dL}(C^{\bullet}_{\breve{X}})$ and $H_{\partial +\overline{\partial}} (\hat{C}^{\bullet}_{X})$ are generally infinite-dimensional, as a striaghtforward consequence of Theorem \ref{Fourier-Mukai_iso}, we have:

\begin{theo}\label{Fourier_Mukai_sym}
	The Fourier-Mukai transform induces the following isomorphism for all $k$:
	\[
	H_{B,\partial + \overline{\partial}}^{k}(\hat{C}_{X}^{\bullet}) \cong H^{k}_{B,\d + \dL}(C_{\breve{X}}^{\bullet}).
	\]
\end{theo}

\ 
	
	\section{The Tseng-Yau cohomology of Solvmanifolds}
	
Consider the semi-direct product $G \ltimes_{\rho^{*}} \RR^{n}$, and let $\{e_1, \ldots, e_{2n}\}$ be a basis of left-invariant $1$-forms for the Lie algebra $\mathfrak{g} \ltimes_{\d_e \rho^*} \RR^{n}$. Then the exterior derivatives are given by:
\begin{align}
	\d e_i &= -\sum_{1\le j<k\le n} C^{i}_{j,k}\, e_j\wedge e_k, &(1\le i\le n), \\
	\d e_{i+n} &= -\sum_{j,k=1}^{n} (\d_{e}\rho^{\ast}(X_j))_{i,k}\, e_j\wedge e_{k+n}, &(1\le i\le n).
\end{align}

The construction of the non-Kähler SYZ mirror pair in Section~\ref{affine} yields a symplectic manifold $(\breve{X}, \breve{\omega}, \breve{\mathsf{\Omega}})$, which is a solvmanifold obtained as the quotient of $G \ltimes_{\rho^{*}} \RR^{n}$ by the lattice $\Gamma \ltimes_{\rho^{*}} \Xi^{*}$. From the definition of the product in $G \ltimes_{\rho^{*}} \RR^{n}$, the cohomology classes in $H_{\d+\dL}(\breve{X})$ can be represented by left-invariant forms, and the $(p,q)$-decomposition of $H_{\d+\dL}(\breve{X})$ can be computed explicitly from this basis of left-invariant $1$-forms. Theorem \ref{thm:computing-cohomology} below provides a method to compute the cohomology $H_{\d+\dL}(\breve{X})$. Applying the Fourier-Mukai transform, one can then compute the corresponding cohomology on the mirror manifold $X$.

\begin{theo_with_name}{Angella, Kasuya~\cite[Theorem 3.2]{AngellaKasuya2019}} \label{thm:computing-cohomology}
	Let $\breve{X} = \Gamma \backslash G$ be a solvmanifold endowed with a $G$-left-invariant symplectic structure $\breve{\omega}$. Suppose that the natural inclusion map 
	\[
	H^\bullet_{\mathrm{dR}}(\mathfrak{g}; \RR) \longrightarrow H^\bullet_{\mathrm{dR}}(\breve{X}; \RR)
	\]
	is an isomorphism. Then, the natural inclusion map
	\[
	\iota \colon H^\bullet_{\d+ \dL}(\mathfrak{g}; \RR) \longrightarrow H^\bullet_{\d + \dL}(\breve{X}; \RR)
	\]
	is also an isomorphism.
\end{theo_with_name}

In this section, we illustrate this result in two distinct cases. For almost abelian manifolds, we provide an explicit description of their Tseng-Yau cohomology. In the case of the generalized Heisenberg group, the computation is more involved and we only present the necessary components for the complete calculation.

\

	\subsection{Almost abelian manifolds}

Let $G = \RR \ltimes_{e^{x_1 A}} \RR^{n-1}$ be an almost abelian Lie group, where $\mathrm{tr}(A) = 0$. Assume that there exists $t_0 \in \RR \setminus \{0\}$ such that $e^{t_0 A}$ is conjugate to an integer matrix. Consider $G$ equipped with a complete left-invariant affine structure, given by the affine representation $\alpha \colon G \to \mathrm{Aff}(\RR^{n})$ with linear part $\rho(x_{1},\ldots , x_{n}) = \mathrm{diag}(1 , e^{x_1A})$. As shown in Appendix~\ref{ap:lattice}, there exists a lattice $\Gamma$ such that the quotient $B = G / \Gamma$ has an integral affine structure.

\vspace{1em}

Let $(X,\mathsf{\Omega} ,\omega) \to B$ and $(\breve{X}, \breve{\omega},\breve{\mathsf{\Omega}}) \to B$ be a supersymmetric mirror pair of type IIB and IIA, respectively, as provided in Theorem~\ref{mirrorsolvmanifolds}. By Theorem~\ref{thm:dedulvan}, the group $G \ltimes_{\rho^{*}} \RR^{n}$ acts transitively on $\breve{X}$ with stabilizer $\Gamma \ltimes_{\rho^{*}} \Xi$. Note that, by the definition of $\rho^{*}$, the quotient $(G \ltimes_{\rho^{*}} \RR^{n}) / (\Gamma \ltimes_{\rho^{*}} \Xi)$ is a solvmanifold. We now proceed to compute the cohomology $H_{\d + \dL}(\mathfrak{g} \ltimes_{\d _e \rho^{*}} \RR^{n}; \RR)$.

Let $\{\breve{e}_1, \ldots , \breve{e}_{2n}\}$ be a basis of left-invariant $1$-forms. The differentials are determined by:
\begin{align}
	\d \breve{e}_i &= -\sum_{k=2}^{n} a_{ki} \breve{e}_1 \wedge \breve{e}_k, &(1 \le i \le n), \\
	\d \breve{e}_{i+n} &= \sum_{k=2}^{n} a_{ik} \breve{e}_1 \wedge \breve{e}_{i+n}, &(1\le i \le n),
\end{align}
where $\d \breve{e}_1 = \d \breve{e}_{n+1} = 0$. Consider the form $\breve{e}_{I} \wedge \breve{e}_{J+n} \in \bigwedge^{s} \left( \mathfrak{g} \ltimes_{\d_e \rho^{*}} \RR^{n} \right)$, where $I = (i_1, \ldots, i_m)$ and $J = (j_1 ,\ldots, j_r)$ are index sets such that $m + r = s$. Expanding the differential, we find:
\begin{align}
	\d (\breve{e}_I \wedge \breve{e}_{J+n}) 
	&= -\sum_{l=1}^{m}\sum^{n}_{k=2} a_{i_l k} \breve{e}_1 \wedge \breve{e}_k \wedge \breve{e}_{i_1} \wedge \dots \wedge \widehat{\breve{e}_{i_l}} \wedge \dots \wedge \breve{e}_{i_m} \wedge \breve{e}_{j_1 +n} \wedge \dots \wedge \breve{e}_{j_r +n} \\
	&\quad + \sum^{r}_{l=1} \sum^{n}_{k=2} a_{k j_l} \breve{e}_1 \wedge  \breve{e}_{i_1} \wedge \dots \wedge \breve{e}_{i_{m}} \wedge \breve{e}_{k+n} \wedge \breve{e}_{j_1 + n} \wedge \dots \wedge \widehat{\breve{e}_{j_l + n}} \wedge \dots \wedge \breve{e}_{j_r + n} \\
	&= - \sum^{m}_{l=1}\sum_{k \in (I\cup \{1\})^{c}} a_{i_l k} \breve{e}_1 \wedge \breve{e}_k \wedge \breve{e}_{I\setminus\{i_l\}} \wedge \breve{e}_{J+n} \\
	&\quad + \sum^{r}_{l=1} \sum_{k \in (J\cup \{1\})^{c}} a_{k j_l} \breve{e}_1 \wedge \breve{e}_{I} \wedge \breve{e}_{k+n} \wedge \breve{e}_{J\setminus\{j_l\}+n} \\
	&\quad + \left( -\sum_{i \in I\setminus\{1\}} a_{ii} + \sum_{j \in J\setminus\{1\}} a_{jj} \right) \breve{e}_1 \wedge \breve{e}_I \wedge \breve{e}_{J+n}.
\end{align}
Consequently, if $1 \notin I$, the condition $\d (\breve{e}_{I} \wedge \breve{e}_{J+n}) = 0$ is satisfied if and only if:
\begin{align}
	a_{i k} &= 0 \quad \text{for all } i \in I, \; k \in (I\cup\{1\})^{c}, \\
	a_{k j} &= 0 \quad \text{for all } j \in J, \; k \in (J\cup\{1\})^{c}, \\
	\sum_{i \in I} a_{ii} &= \sum_{j \in J\setminus\{1\}} a_{jj}.
\end{align}
If $1 \in I$, the differential vanishes identically, i.e., $\d ( \breve{e}_{I}\wedge \breve{e}_{J+n} ) = 0$. 

Recalling that $\breve{\omega} = \sum_{k=1}^{n} \breve{e}_{k} \wedge \breve{e}_{k+n}$, the dual Lefschetz operator acts as:
\[
\Lambda(\breve{e}_{I}\wedge \breve{e}_{J+n}) = \sum_{\sigma \in I \cap J} (-1)^{\xi_\sigma}\breve{e}_{I\setminus \{\sigma\}} \wedge \breve{e}_{J\setminus \{\sigma\} + n}
\]
for appropriate signs $(-1)^{\xi_\sigma}$. Differentiating this yields:
\begin{align}
	\d \Lambda(\breve{e}_{I}\wedge \breve{e}_{J+n}) 
	&= \sum_{\sigma \in I \cap J} (-1)^{\xi_\sigma} \d \left( \breve{e}_{I\setminus \{\sigma\}} \wedge \breve{e}_{J\setminus \{\sigma \} +n } \right)  \\
	&= -\sum_{\sigma \in I \cap J} \sum_{i \in I\setminus\{\sigma\}}\sum_{k \in (I\setminus \{\sigma ,1 \})^{c}} (-1)^{\xi_{\sigma}} a_{i k} \breve{e}_1 \wedge \breve{e}_k \wedge \breve{e}_{I\setminus \{ i, \sigma\}} \wedge \breve{e}_{J\setminus \{ \sigma\}+n} \\
	&\quad + \sum_{\sigma \in I \cap J} \sum_{j \in J\setminus \{\sigma\}} \sum_{k \in (J\setminus \{\sigma,1 \})^{c}} (-1)^{\xi_{\sigma}}a_{k j} \breve{e}_1 \wedge \breve{e}_{I\setminus \{\sigma\}} \wedge \breve{e}_{k+n} \wedge \breve{e}_{J\setminus \{j,\sigma\}+n} \\
	&\quad + \sum_{\sigma \in I \cap J} (-1)^{\xi_{\sigma}}\left( -\sum_{i \in I \setminus \{\sigma,1\}} a_{ii} + \sum_{j \in J\setminus \{\sigma,1 \}} a_{jj} \right) \breve{e}_1 \wedge \breve{e}_{I\setminus \{\sigma\}} \wedge \breve{e}_{J\setminus \{\sigma\}+n}.
\end{align}
For $1 \in I$ and $1 \in J$, the vanishing condition $\d \Lambda(\breve{e}_I \wedge \breve{e}_{J+n}) = 0$ simplifies to:
\begin{align}
	a_{ik} &= 0 \quad \text{for all } i \in I \setminus\{1\}, \; k \in (I \setminus \{1\})^{c}, \\
	a_{kj} &= 0 \quad \text{for all } j \in J \setminus\{1\}, \; k \in (J \setminus \{1\})^{c}, \\
	\sum_{i \in I\setminus\{1\}}a_{ii} &= \sum_{j \in J\setminus\{1\}} a_{jj}.
\end{align}

We explicitly characterize the kernel of the differential as follows:
\begin{align}
	\ker(\d) &= \left\{ \breve{e}_I \wedge \breve{e}_{J+n} \;\middle|\; 
	\begin{aligned} 
		&I \subset \{2, \ldots, n\}, \, J \subset \{1, \ldots , n\}, \, |I|+|J| = k, \\
		&a_{i k} = 0 \text{ for } i \in I, \; k \in (I\cup \{1\})^{c}, \\ 
		&a_{k j} = 0 \text{ for } j \in J, \; k \in (J\cup\{1\})^{c}, \\ 
		&\sum_{i \in I} a_{ii} = \sum_{j \in J} a_{jj}
	\end{aligned} 
	\right\} \\
	&\quad \cup \left\langle \breve{e}_1 \wedge \breve{e}_I \wedge \breve{e}_{J+n} \;\middle|\;
	\begin{aligned} 
		&I \subset \{2, \ldots , n\}, \, J \subset \{1,\ldots, n\}, \\ 
		&|I|+|J| = k-1
	\end{aligned} 
	\right\rangle.
\end{align}

To compute the cohomology, we assume that $A$ is a diagonal matrix. Under this assumption, for $1 \in J$, the operator acts as:
\begin{align}
	\d \Lambda \d(\breve{e}_I \wedge \breve{e}_{J+n}) 
	&= \left( -\sum_{i \in I} a_{ii} + \sum_{j \in J} a_{jj} \right) \d \left( \breve{e}_I \wedge \breve{e}_{J\setminus\{1\}+n} \right)   \\
	&= \left( -\sum_{i \in I} a_{ii} + \sum_{j \in J} a_{jj} \right)^{2} \breve{e}_1 \wedge \breve{e}_I \wedge \breve{e}_{J\setminus\{1\}+n}.
\end{align}
Recalling that $H_{\d+\dL}^{k}(\breve{X}) = \frac{\ker(\d)\cap \ker(\d\Lambda)}{\im(\d \Lambda \d)}$, we arrive at the following explicit characterization of the cohomology:
\[
H_{\d +\dL }^{k}(\breve{X}) \cong \left\langle \breve{e}_I \wedge \breve{e}_{J+n} \;\middle|\; 
\begin{aligned} 
	&I,J \subset \{1,\ldots, n\}, \, |I|+|J| = k \\
	&\sum_{i \in I \setminus \{1\} } a_{ii} = \sum_{j \in J \setminus\{1\}} a_{jj} 
\end{aligned} \right\rangle.
\]
The corresponding $(p,q)$-decomposition is then given by:
\[
H^{p,q}_{\d +\dL,B}(\breve{X}) \cong \left\langle \breve{e}_I \wedge \breve{e}_{J+n} \;\middle|\; 
\begin{aligned} 
	&I,J \subset \{1,\ldots, n\}, \, |I| = p, \, |J| = q \\
	&\sum_{i \in I \setminus \{1\} } a_{ii} = \sum_{j \in J \setminus\{1\}} a_{jj} 
\end{aligned} \right\rangle.
\]
Finally, by Theorem~\ref{mirror_Yau_cohomologies}, the Bott-Chern cohomology on the complex mirror side $X$ is:
\[
H^{p,q}_{\partial +\overline{\partial},B}(X) \cong \left\langle e_I \wedge \overline{e}_{J} \;\middle|\; 
\begin{aligned} 
	&I,J \subset \{1,\ldots, n\}, \, |I| = n - p, \, |J| = q \\
	&\sum_{i \in I \setminus \{1\} } a_{ii} = \sum_{j \in J \setminus\{1\}} a_{jj} 
\end{aligned} \right\rangle_{\CC}.
\]

\vspace{1em}

\begin{remark}
	It is possible to compute the cohomology on the complex side explicitly. This is the content of Theorem \ref{thm:complex-side}, whose proof is omitted:
	
	\begin{theo}\label{thm:complex-side}
		Let $G = \RR \ltimes_{e^{x_1 A}} \RR^{n-1}$ be an almost abelian Lie group, where $A$ is a diagonalizable matrix with $\mathrm{tr}(A)=0$. Let $\rho \colon G \to \mathrm{GL}(n,\RR)$ be such that $\rho(x_1, \ldots, x_n) = \mathrm{diag}(1,e^{x_1A})$. Assume the existence of a lattice $\Gamma \ltimes_\rho \Xi$ such that $(G \ltimes_\rho \RR^{n}) / (\Gamma \ltimes_\rho \Xi)$ is a solvmanifold equipped with the canonical left-invariant complex structure. Then
		\[
		H^{p,q}_{\partial + \overline{\partial}}(X) \cong \mathrm{span}_{\CC} \left\{ Y^{*}_{I} \wedge \overline{Y^{*}_{J}} \;\middle|\; 
		\begin{aligned} 
			&|I| = p, \, |J| = q \\ 
			&\sum_{k \in I\setminus \{1\}} \lambda_{k} = -\sum_{j \in J \setminus \{1\}} \lambda_j 
		\end{aligned} \right\}
		\]
		where $\lambda_{2}, \ldots, \lambda_{n}$ are the eigenvalues of $A$; $Y_{2}, \ldots, Y_{n}$ are the associated eigenvectors and $\d z_1 = Y^{*}_1$.
	\end{theo}
\end{remark}
	
	\subsection{Generalized Heisenberg group}
	
Let $\mathcal{H}_{2n+1}(\RR)$ be the generalized Heisenberg group. It admits a left-invariant affine structure given by the developing map:
\begin{align}
	D \colon \mathcal{H}_{2n+1}(\RR) & \longrightarrow \RR^{2n+1} \\
	\begin{bmatrix}
		1 & \overline{a} & c \\
		0 & I_{n} & \overline{b} \\
		0 & 0 & 1
	\end{bmatrix} & \longmapsto (\overline{a}, c , \overline{b}).
\end{align}
Furthermore, $\mathcal{H}_{2n+1}(\RR)$ contains a lattice $\Gamma$ such that the affine structure induced on the quotient $\mathcal{H}_{2n+1}(\RR)/\Gamma$ is integral. The linear part of the affine representation is explicitly given by:
\begin{align}
	\rho \colon \mathcal{H}_{2n+1}(\RR) & \longrightarrow \mathrm{GL}(\RR^{2n+1}) \\
	\begin{bmatrix}
		1 & \overline{a} & c \\
		0 & I_{n} & \overline{b} \\
		0 & 0 & 1
	\end{bmatrix} & \longmapsto \begin{bmatrix}
		I_n & 0 & 0 \\
		0 & 1 & \overline{a} \\
		0 & 0 & I_n
	\end{bmatrix}.
\end{align}

Let $\{X_1, \ldots, X_n , T , Y_1, \ldots , Y_n\}$ be a basis for the Lie algebra of $\mathcal{H}_{2n+1}(\RR)$. We then have:
\[
\rho_{*}(X_i) = \begin{bmatrix}
	0_{n} & 0 & 0 \\
	0 & 0 & E_i \\
	0 & 0 & 0_{n}
\end{bmatrix} = E_{n+1, n+1+i}
\]
where $E_i$ is a row vector of zeros with a $1$ in the $i$-th entry, and $E_{i,j}$ denotes a matrix with a $1$ in the $(i,j)$-th entry and zeros everywhere else. Additionally, $\rho_{*}(Y_i) = \rho_{*}(T) = 0$. 

Denoting by $\rho^{*}$ the dual representation of $\rho$, we find:
\[
\d_{e} \rho^{*}(X_i) = -(\d_e \rho(X_i))^{T}
\]
and $\d_{e} \rho^{*}(Y_i) = \d_e \rho^{*}(T) = 0$ for all $1\le i \le n$. By Lemma~\ref{lemmaliealgebra}, the structure constants of $G \ltimes_{\rho^{*}}\RR^{n}$ are given by:
\[
c^{p}_{qr} =
\begin{cases}
	C^{p}_{qr}, & 1\le p,q,r\le 2n+1, \\[4pt]
	-(\rho_{\ast}(S_q))_{r-n,\, p-n}, & 1\le q\le 2n+1,\; 2n+2\le p,r\le 4n+2, \\[4pt]
	(\rho_{\ast}(S_r))_{q-n, \, p-n}, & 2n+2\le q\le 4n+2,\; 1\le r\le 2n+1,\; 2n+2\le p\le 4n+2, \\[4pt]
	0, & \text{otherwise},
\end{cases}
\]
where $\{S_1,\ldots, S_{2n+1}\}$ is a basis for $\mathfrak{h}$, defined as $S_{i} = X_i$ for $1 \le i \le n$, $S_{n+1} = T$, and $S_{n+1+i} = Y_i$ for $1 \le i \le n$. Because $\mathrm{ad}_{X}$ has real eigenvalues for all $X$, the semi-direct product $G \ltimes_{\rho^{*}} \RR^{n}$ is completely solvable.

Let $(X,\mathsf\Omega ,\omega) \to B$ and $(\breve{X}, \breve{\omega},\breve{\mathsf\Omega}) \to B$ be a supersymmetric mirror pair of type IIB and IIA, respectively, as provided in Theorem~\ref{mirrorsolvmanifolds}. By Theorem~\ref{thm:dedulvan}, $G \ltimes_{\rho^{*}} \RR^{n}$ acts transitively on $\breve{X}$ with stabilizer $\Gamma \ltimes_{\rho^{*}} \Xi$. Note that, by the definition of $\rho^{*}$, the quotient $(G \ltimes_{\rho^{*}} \RR^{n}) / (\Gamma \ltimes_{\rho^{*}} \Xi)$ is a solvmanifold.

Our goal is to compute $H_{\d + \dL}(\mathfrak{g} \ltimes_{\d _e \rho^{*}} \RR^{n}; \RR)$. Let $\{\breve{e}_1, \ldots , \breve{e}_{4n+2}\}$ be a basis of left-invariant $1$-forms. The structure constants of $\mathcal{H}_{2n+1}(\RR)$ dictate:
\[
C^{i}_{j,k} = \begin{cases}
	1, & i = n+1, \, 1\le j \le n ,\, k = j + n+1 , \\
	-1, & i = n+1, \, 1 \le k \le n, \, j = k+n+1, \\
	0, & \text{otherwise}.
\end{cases}
\]
Thus, $\d \breve{e}_i = \d \breve{e}_{i+n+1} = 0$ for all $1 \le i \le n$, and:
\[
\d \breve{e}_{n+1} = - \sum_{i=1}^{n} \breve{e}_i \wedge \breve{e}_{i+n+1}.
\]
Recall that:
\[
\d \breve{e}_{i+2n+1} = \sum_{j,k=1}^{2n+1} (\rho_{\ast}(S_j))_{k,i}\, \breve{e}_j\wedge \breve{e}_{k+2n+1} \qquad (1\le i\le 2n+1).
\]
This simplifies to:
\[
\d \breve{e}_{i+2n+1} = \sum^{n}_{j=1}\sum_{k=1}^{2n+1}(\rho_{\ast}(X_j))_{k,i}\, \breve{e}_j\wedge \breve{e}_{k+2n+1} \qquad (1\le i\le 2n+1).
\]
Noting that $(\rho_{*}(X_j))_{k,i} = \delta_{k,n+1}\delta_{i,j+n+1}$, we obtain:
\[
\d \breve{e}_{i+2n+1} = \sum^{n}_{j=1}\sum_{k=1}^{2n+1} \delta_{k,n+1}\delta_{i,j+n+1} \, \breve{e}_j\wedge \breve{e}_{k+2n+1} = \sum^{n}_{j=1} \delta_{i,j+n+1} \, \breve{e}_{j} \wedge \breve{e}_{3n+2} \qquad (1\le i\le 2n+1).
\]
Consequently, we have:
\begin{align}
	\d \breve{e}_{i+2n+1} &= 0, &(1\le i \le n+1), \\
	\d \breve{e}_{i+3n+2} &= \breve{e}_{i} \wedge \breve{e}_{3n+2}, &(1\le i \le n).
\end{align}

Let $\breve{\omega} = \sum_{i=1}^{2n+1}\breve{e}_{i} \wedge \breve{e}_{i+2n+1}$. Consider index sets $I = (i_1 < \ldots< i_m)$ and $J = (j_1 < \ldots < j_{r})$ with $I,J \subset \{1, \ldots , 2n+1\}$. Suppose first that $n+1 \notin I$. Then:
\[
\d(\breve{e}_{I} \wedge \breve{e}_{J+2n+1}) = (-1)^{|I|} \breve{e}_I \wedge \d \breve{e}_{J+2n+1}.
\]
Denoting $W = \{1,\ldots, n+1\}$, we can expand this as:
\begin{align}
	\d(\breve{e}_{I} \wedge \breve{e}_{J+2n+1}) 
	&= (-1)^{|I|+|J \cap W|} \breve{e}_{I} \wedge \breve{e}_{(J\cap W) +2n+1} \wedge \d \breve{e}_{(J \setminus W) +2n+1} \\
	&= \sum_{j \in J \setminus W} (-1)^{|I|+|J \cap W|} \breve{e}_{I} \wedge \breve{e}_{(J\cap W) +2n+1} \wedge \breve{e}_j \wedge \breve{e}_{3n+2} \wedge \breve{e}_{(J\setminus (W\cup \{j\})) + 2n+1} \\
	&= \sum_{j \in J \setminus W}  \breve{e}_j\wedge\breve{e}_{I} \wedge \breve{e}_{(J\cap W) +2n+1} \wedge \breve{e}_{3n+2} \wedge \breve{e}_{(J\setminus (W\cup \{j\})) + 2n+1} \\
	&= \sum_{j \in J \setminus W \cap I^{c}}  \breve{e}_j\wedge\breve{e}_{I} \wedge \breve{e}_{(J\cap W) +2n+1} \wedge \breve{e}_{3n+2} \wedge \breve{e}_{(J\setminus (W\cup \{j\})) + 2n+1} \\
	&= \sum_{j \in J \cap I^{c} \cap W^{c} }\breve{e}_j\wedge\breve{e}_{I} \wedge \breve{e}_{(J\cap W) +2n+1} \wedge \breve{e}_{3n+2} \wedge \breve{e}_{(J\setminus (W\cup \{j\})) + 2n+1}.
\end{align}

Suppose instead that $n+1 \in I$. We can decompose $I = I_1 \cup \{n+1\} \cup (I_2 + n+1)$, where $I_1, I_2 \subset \{1,\ldots , n\}$. This yields:
\begin{align}
	\d(\breve{e}_{I} \wedge \breve{e}_{J+2n+1}) 
	&= \sum_{i=1}^{n}  \breve{e}_{i}\wedge \breve{e}_{I_1} \wedge \breve{e}_{i+n+1}\wedge \breve{e}_{I_2+n+1}\wedge \breve{e}_{J+2n+1} \\
	&\quad + \sum_{j \in J \setminus W}  \breve{e}_j\wedge\breve{e}_{I} \wedge \breve{e}_{(J\cap W) +2n+1} \wedge \breve{e}_{3n+2} \wedge \breve{e}_{(J\setminus (W\cup \{j\})) + 2n+1} \\
	&= \sum_{i \in (I_1 \cup I_2)^{c}} \breve{e}_i \wedge \breve{e}_{I_1} \wedge \breve{e}_{i+n+1} \wedge \breve{e}_{I_2 +n+1} \wedge \breve{e}_{J+2n+1} \\
	&\quad + \sum_{j \in J \cap I^{c} \cap W^{c}}  \breve{e}_j\wedge\breve{e}_{I} \wedge \breve{e}_{(J\cap W) +2n+1} \wedge \breve{e}_{3n+2} \wedge \breve{e}_{(J\setminus (W\cup \{j\})) + 2n+1}.
\end{align}

Applying the dual Lefschetz operator $\Lambda$ and taking the exterior derivative yields:
\[
\d \Lambda(\breve{e}_{I} \wedge \breve{e}_{J+2n+1}) = \sum_{\sigma \in I \cap J} \d \left( \breve{e}_{I\setminus \{\sigma\}} \wedge \breve{e}_{(J\setminus \{\sigma\}) +2n+1} \right).
\]
If $n+1 \notin I$, this expands to:
\begin{align}
	\d \Lambda(\breve{e}_{I} \wedge \breve{e}_{J+2n+1})  
	&=  \sum_{\sigma \in I \cap J} \sum_{j \in (J\cap I^{c}\cap W^{c}) \setminus \{\sigma\}} \breve{e}_j \wedge \breve{e}_{I\setminus \{\sigma\}} \wedge \breve{e}_{(J\cap W) \setminus\{\sigma\} +2n+1} \wedge \breve{e}_{3n+2} \wedge \breve{e}_{J\setminus (W\cup \{j, \sigma\}) + 2n+1} \\
	&= \sum_{\sigma \in I \cap J} \sum_{j \in J\cap I^{c}\cap W^{c} }  \breve{e}_j\wedge\breve{e}_{I\setminus \{\sigma\}} \wedge \breve{e}_{(J\cap W) \setminus\{\sigma\} +2n+1} \wedge \breve{e}_{3n+2} \wedge \breve{e}_{J\setminus (W\cup \{j,\sigma\}) + 2n+1}.
\end{align}
If $n+1 \in I$:
\begin{align}
	\d \Lambda(\breve{e}_{I} \wedge \breve{e}_{J+2n+1}) 
	&= \sum_{\sigma \in I \cap J}\sum_{i \in (I_1 \cup I_2 \cup \{\sigma\})^{c}} \breve{e}_i \wedge \breve{e}_{I_1\setminus \{\sigma\}} \wedge \breve{e}_{i+n+1} \wedge \breve{e}_{(I_2\setminus\{\sigma\}) +n+1} \wedge \breve{e}_{(J\setminus\{\sigma\})+2n+1} \\
	&\quad + \sum_{\sigma \in I \cap J}\sum_{j \in J \cap I^{c} \cap W^{c}}  \breve{e}_j\wedge\breve{e}_{I\setminus\{\sigma\}} \wedge \breve{e}_{(J\cap W) \setminus \{\sigma\} +2n+1} \wedge \breve{e}_{3n+2} \wedge \breve{e}_{J\setminus (W\cup \{j,\sigma\}) + 2n+1}.
\end{align}

Furthermore, if $n+1 \notin I$, the composition $\d\Lambda\d$ vanishes:
\begin{align}
	\d \Lambda \d (\breve{e}_I \wedge \breve{e}_{J+2n+1})  
	&= \d \Lambda \left(  \sum_{j \in J \cap I^{c} \cap W^{c}}  \breve{e}_j\wedge\breve{e}_{I} \wedge \breve{e}_{(J\cap W) +2n+1} \wedge \breve{e}_{3n+2} \wedge \breve{e}_{(J\setminus (W\cup \{j\})) + 2n+1} \right)  \\
	&= \d \left(  \sum_{j \in J \cap I^{c} \cap W^{c}} \sum_{\sigma \in I \cap J}   \breve{e}_j\wedge\breve{e}_{I} \wedge \breve{e}_{(J\cap W)\setminus\{\sigma\} +2n+1} \wedge \breve{e}_{3n+2} \wedge \breve{e}_{(J\setminus (W\cup \{j, \sigma\})) + 2n+1} \right) \\
	&= \sum_{j \in J \cap I^{c} \cap W^{c}} \sum_{\sigma \in I \cap J}  \sum_{k \in J\cap W^{c} \setminus \{j\} \cap (I^{c} \cup \{\sigma \})  }  \breve{e}_{k}\wedge \breve{e}_j\wedge\breve{e}_{I \setminus \{\sigma\}} \wedge \breve{e}_{(J\cap W)\setminus\{\sigma\} +2n+1} \\
	&\quad \quad \wedge \breve{e}_{3n+2}\wedge \breve{e}_{3n+2} \wedge \breve{e}_{(J\setminus (W\cup \{j, \sigma,k\})) + 2n+1} \\
	&= 0.
\end{align}

Finally, if $n+1 \in I$ and $n+1 \notin J$, we expand the $\d\Lambda\d$ operation into three distinct parts:
\begin{align}
	\d \Lambda \d &(\breve{e}_I \wedge \breve{e}_{J+2n+1}) \\
	&= \d \Lambda \left(   \sum_{i \in (I_1 \cup I_2)^{c}} \breve{e}_i \wedge \breve{e}_{I_1} \wedge \breve{e}_{i+n+1} \wedge \breve{e}_{I_2 +n+1} \wedge \breve{e}_{J+2n+1} \right) \\
	&\quad + \d \Lambda \Bigg( \sum_{j \in J \cap I^{c} \cap W^{c}}  \breve{e}_j \wedge \breve{e}_{I} \wedge \breve{e}_{(J\cap W) +2n+1} \wedge \breve{e}_{3n+2} \wedge \breve{e}_{(J\setminus (W\cup \{j\})) + 2n+1} \Bigg) \\
	&=  \sum_{i \in (I_1 \cup I_2)^{c}} \sum_{\sigma \in (I\cup \{i,i+n+1\})\cap J } \d \Bigg( \breve{e}_{(I_1 \cup I_2 +n+1 \cup \{i,i+n+1\}) \setminus\{\sigma\}} \wedge \breve{e}_{(J\setminus \{\sigma\}) +2n+1} \Bigg)  \\
	&\quad +   \sum_{j \in J \cap I^{c} \cap W^{c}} \sum_{\sigma \in I \cap ((J \cup \{n+1\})\setminus\{j\})} \d \Bigg( \breve{e}_{j} \wedge\breve{e}_{I \setminus\{\sigma\}} \wedge \breve{e}_{(J \cup \{n+1\})\setminus \{\sigma,j\} +2n+1} \Bigg)\\
	&= \sum_{i \in (I_1 \cup I_2)^{c}} \sum_{\sigma \in (I\cup \{i,i+n+1\})\cap J }\sum_{j \in J \setminus (W \cup \{\sigma\})} \breve{e}_j \\
	&\qquad \wedge \breve{e}_{(I_1 \cup I_2 +n+1 \cup \{i,i+n+1\}) \setminus\{\sigma\}} \wedge \breve{e}_{(J\cup \{n+1\}) \setminus\{\sigma,j\} + 2n+1}  \\
	&\quad + \sum_{j \in J \cap I^{c} \cap W^{c}} \sum_{\sigma \in I \cap ((J \cup \{n+1\})\setminus\{j\})} \left( \d \breve{e}_{j} \wedge\breve{e}_{I \setminus\{\sigma\}} \right) \wedge \breve{e}_{(J \cup \{n+1\})\setminus \{\sigma,j\}+2n+1}  \\
	&\quad + (-1)^{|I|} \sum_{j \in J \cap I^{c} \cap W^{c}} \sum_{\sigma \in I \cap ((J \cup \{n+1\})\setminus\{j\})} \breve{e}_{j} \wedge \breve{e}_{I\setminus\{\sigma\}} \\
	&\qquad \wedge \d \left( \breve{e}_{(J \cup \{n+1\})\setminus \{\sigma,j\}+2n+1} \right).
\end{align}

We can simplify the final two summations:
\begin{align}
	&\sum_{j \in J \cap I^{c} \cap W^{c}} \sum_{\sigma \in I \cap ((J \cup \{n+1\})\setminus\{j\})} \left( \d \breve{e}_{j} \wedge\breve{e}_{I \setminus\{\sigma\}} \right) \wedge \breve{e}_{(J \cup \{n+1\})\setminus \{\sigma,j\}+2n+1}  \\
	&\quad +(-1)^{|I|}\sum_{j \in J \cap I^{c} \cap W^{c}} \sum_{\sigma \in I \cap ((J \cup \{n+1\})\setminus\{j\})} \breve{e}_{j} \wedge \breve{e}_{I\setminus\{\sigma\}} \\
	&\qquad \wedge \d \left( \breve{e}_{(J \cup \{n+1\})\setminus \{\sigma,j\}+2n+1} \right) \\
	&=  \sum_{j \in J \cap I^{c} \cap W^{c}} \sum_{\sigma \in I \cap (J \setminus\{j\})} \sum_{i=1}^{n}  \breve{e}_{(I_1 \cup I_2+n+1\cup\{j,i,i+n+1\}) \setminus\{\sigma\}} \\
	&\qquad \wedge \breve{e}_{(J \cup \{n+1\})\setminus \{\sigma,j\}+2n+1}     \\
	&\quad +(-1)^{|I|}\sum_{j \in J \cap I^{c} \cap W^{c}} \breve{e}_{j} \wedge \breve{e}_{I\setminus\{n+1\}} \wedge \d \left( \breve{e}_{(J\setminus \{j\})+2n+1} \right) \\
	&= \sum_{j \in J \cap I^{c} \cap W^{c}} \sum_{\sigma \in I \cap (J \setminus\{j\})} \sum_{i=1}^{n}  \breve{e}_{(I_1\cup I_2+n+1\cup\{j,i,i+n+1\}) \setminus\{\sigma\}} \\
	&\qquad \wedge \breve{e}_{(J \cup \{n+1\})\setminus \{\sigma,j\}+2n+1}     \\
	&\quad +(-1)^{|I|}\sum_{j \in J \cap I^{c} \cap W^{c}} \sum_{k \in J\setminus(W\cup\{j\})} \breve{e}_{k} \wedge \breve{e}_{j} \wedge \breve{e}_{I\setminus\{n+1\}} \wedge \breve{e}_{(J \cup \{n+1\})\setminus \{j\}+2n+1}.
\end{align}

Substituting this back, we arrive at the closed-form expression for $n+1 \in I$ and $n+1 \notin J$:
\begin{align}
	\d \Lambda \d &(\breve{e}_I \wedge \breve{e}_{J+2n+1}) \\
	&= \sum_{i \in (I_1 \cup I_2)^{c}} \sum_{\sigma \in (I\cup \{i,i+n+1\})\cap J } \sum_{j \in J \setminus (W \cup \{\sigma\})} \breve{e}_j \\
	&\qquad \wedge \breve{e}_{(I_1 \cup I_2+n+1 \cup \{i,i+n+1\}) \setminus\{\sigma\}} \wedge \breve{e}_{(J\cup \{n+1\}) \setminus\{\sigma,j\} + 2n+1} \\
	&\quad + \sum_{j \in J \cap I^{c} \cap W^{c}} \sum_{\sigma \in I \cap (J \setminus\{j\})} \sum_{i \in (I_1 \cup I_2 \cup\{j-(n+1)\} \setminus \{\sigma\})^{c}} \\
	&\qquad \breve{e}_{(I_1 \cup I_2+n+1\cup\{j,i,i+n+1\}) \setminus\{\sigma\}} \wedge \breve{e}_{(J \cup \{n+1\})\setminus \{\sigma,j\}+2n+1}.
\end{align}

	\


	\bibliographystyle{amsplain} 
	
	\bibliography{main}

	\appendix
	
	\section{Existence of lattice}
	\label{ap:lattice}
	
In Example \ref{almostabelianmirror}, we assumed the existence of a lattice $\Gamma$ with the desired properties. It turns out that determining when such lattices exist in this context (and in broader contexts) is an active area of research. In dimension three, lattices related to affine structures are called affine crystallographic groups; a comprehensive discussion on this topic can be found in \cite{Fried_1980}. For nilpotent and solvable manifolds, the existence of such lattices was studied in \cite{Bock_2010} and \cite{Andrada_2016}.

\vspace{1em}

Let $G$ be a simply connected Lie group. A complete left-invariant affine structure on $G$ is equivalent to an affine representation $\alpha \colon G \to \operatorname{Aff}(\RR^n)$. Let $\rho \colon G \to \mathrm{GL}(n,\RR)$ be the linear part of $\alpha$, given by:
\[
\rho(x_1, x_2, \ldots, x_n) = 
\begin{bmatrix}
	1  & 0 \\
	0  & e^{x_1 A}
\end{bmatrix}.
\] 

In the almost abelian case, Bock~\cite{Bock_2010} proved the following theorem (see also \cite[Proposition 2.3]{Andrada_2016}).

\begin{theo_with_name}{Bock~\cite{Bock_2010}, Andrada--Origlia~\cite{Andrada_2016}}
	Let $G = \RR \ltimes_{\phi} \RR^{2n+1}$ be an almost abelian Lie group. Then $G$ admits a lattice if and only if there exists a $t_0 \neq 0$ such that $\phi(t_0)$ is conjugate to an integer matrix. The lattice is given by $\Gamma = t_0 \ZZ \ltimes_\phi P^{-1} \ZZ^{2n+1}$, where $P \phi(t_0)P^{-1}$ is an integer matrix. 
\end{theo_with_name}

\vspace{1em}

Using this theorem, we are able to prove the following result, which asserts the existence of a lattice such that the left-invariant complete affine structure induced on the quotient $G/\Gamma$ is integral.

\begin{theo}\label{thm:existence-lattice}
    Let $G = \RR \ltimes_{e^{x_1 A}} \RR^{2n+1}$ be an almost abelian Lie group. Assume there exists a real number $t_0 \in \RR \setminus \{0\}$ and a matrix $P \in \mathrm{GL}(2n+1,\RR)$ such that $e^{t_0 A}$ is conjugate to an integer matrix via $P$. 
    
    Consider a complete left-invariant affine structure on $G$ given by the affine representation $\alpha \colon G \to \operatorname{Aff}(\RR^{2n+2})$, whose linear part $\rho \colon G \to \mathrm{GL}(2n+2, \RR)$ is given by:
	\[
	\rho(x_1, \ldots, x_{2n+2}) = \begin{bmatrix}
		1  & 0 \\
		0 &  e^{x_1 A}
	\end{bmatrix}.
	\]
	Then, there is a lattice $\Gamma = t_0 \ZZ \ltimes_{e^{x_1 A}} P^{-1}\ZZ^{2n+1}$ such that the affine structure induced on the solvmanifold $G/\Gamma$ is integral.
\end{theo}
\begin{proof}
	Since $e^{t_0 A}$ is conjugate to an integer matrix via $P$, the preceding theorem guarantees the existence of the lattice $\Gamma = t_0 \ZZ \ltimes_{e^{x_1 A}} P^{-1} \ZZ^{2n+1}$.
	
	It remains to prove that the left-invariant affine structure induced on $G/\Gamma$ is integral. Since $\pi_1 (G/\Gamma) \cong \Gamma$, the holonomy of the affine structure is given by the restriction $\restr{\alpha}{\Gamma} \colon \Gamma \to \operatorname{Aff}(\RR^{2n+2})$, whose linear part is $\restr{\rho}{\Gamma}$. To prove that the induced affine structure is integral, we need to show that the image of $\restr{\rho}{\Gamma}$ is conjugate to a subgroup of $\mathrm{GL}(2n+2,\ZZ)$.
	
	Let $(x_1, x_2, \ldots, x_{2n+2}) \in \Gamma$. By the definition of the lattice, the first coordinate must be an integer multiple of $t_0$, meaning $x_1 = k t_0$ for some $k \in \ZZ$. The linear part evaluated on this element is $\rho(k t_0, \ldots) = \mathrm{diag}(1, e^{k t_0 A})$. 
	
	We conjugate this representation using the block diagonal matrix $\Xi = \begin{bmatrix} 1 & 0 \\ 0 & P \end{bmatrix}$:
	\[
	\Xi \, \rho(x_1, \ldots, x_{2n+2}) \, \Xi^{-1} = 
	\begin{bmatrix}
		1 & 0 \\
		0 & P
	\end{bmatrix}
	\begin{bmatrix}
		1 & 0 \\
		0 & e^{k t_0 A}
	\end{bmatrix}
	\begin{bmatrix}
		1 & 0 \\
		0 & P^{-1}
	\end{bmatrix}
	= 
	\begin{bmatrix}
		1 & 0 \\
		0 & P e^{k t_0 A} P^{-1}
	\end{bmatrix}
	= 
	\begin{bmatrix}
		1 & 0 \\
		0 & (P e^{t_0 A} P^{-1})^k
	\end{bmatrix}.
	\]
	Since $P e^{t_0 A} P^{-1}$ is an integer matrix by hypothesis, its $k$-th power is also an integer matrix. Thus, the linear holonomy $\restr{\rho}{\Gamma}$ is conjugate to a matrix in $\mathrm{GL}(2n+2,\ZZ)$. Therefore, the left-invariant affine structure induced by $\alpha$ on $G/\Gamma$ is complete and integral.
\end{proof}

\end{document}